\documentclass[a4paper,onecolumn,superscriptaddress,11pt]{quantumarticle}
\pdfoutput=1
\usepackage[utf8]{inputenc}
\usepackage[english]{babel}
\usepackage[T1]{fontenc}

\usepackage[margin=1in]{geometry}

\usepackage{enumerate}

\usepackage{microtype}

\usepackage{tikz}
\usetikzlibrary{calc,shapes.geometric,math}

\usepackage[colorlinks = true]{hyperref}
\hypersetup{
  pdftitle = {How to combine three quantum states},
  pdfauthor = {Maris Ozols}
}
\definecolor{darkred}  {rgb}{0.5,0,0}
\definecolor{darkblue} {rgb}{0,0,0.5}
\definecolor{darkgreen}{rgb}{0,0.5,0}
\hypersetup{
  urlcolor   = blue,         
  linkcolor  = darkblue,     
  citecolor  = darkgreen,    
  filecolor  = darkred       
}

\usepackage{amsmath,amssymb,amsfonts,amsthm,amstext}

\usepackage{etoolbox}

\usepackage{mathtools}
\mathtoolsset{centercolon}
\makeatletter
\protected\def\tikz@nonactivecolon{\ifmmode\mathrel{\mathop\ordinarycolon}\else:\fi}
\makeatother

\usepackage{cleveref}
\crefname{lemma}{Lemma}{Lemmas}
\crefname{proposition}{Proposition}{Propositions}
\crefname{definition}{Definition}{Definitions}
\crefname{theorem}{Theorem}{Theorems}
\crefname{conjecture}{Conjecture}{Conjectures}
\crefname{corollary}{Corollary}{Corollaries}
\crefname{section}{Section}{Sections}
\crefname{appendix}{Appendix}{Appendices}
\crefname{figure}{Fig.}{Figs.}
\crefname{table}{Table}{Tables}

\usepackage{empheq}
\definecolor{gray}{rgb}{0.95, 0.95, 0.95}
\newcommand*\graybox[1]{\colorbox{gray}{\hspace{2pt}#1\hspace{2pt}}}
\newenvironment{falign}{\empheq[box=\graybox]{align}}{\endempheq}

\newcommand{\ket}[1]{|#1\rangle}
\newcommand{\bra}[1]{\langle#1|}
\newcommand{\braket}[2]{\langle#1|#2\rangle}
\newcommand{\proj}[1]{|#1\rangle\langle#1|}

\newcommand{\x}{\otimes}

\newcommand{\ct}{^{\dagger}}
\newcommand{\tp}{^{\mathsf{T}}}

\DeclarePairedDelimiter{\set}{\lbrace}{\rbrace}
\DeclarePairedDelimiter{\abs}{\lvert}{\rvert}

\DeclarePairedDelimiter{\of}{\lparen}{\rparen}
\DeclarePairedDelimiter{\sof}{\lbrack}{\rbrack}
\DeclarePairedDelimiter{\ip}{\langle}{\rangle}

\renewcommand{\Re}{\operatorname{Re}}
\renewcommand{\Im}{\operatorname{Im}}

\DeclareMathOperator{\spn}{span}

\DeclareMathOperator{\Tr}{Tr}

\newcommand{\mx}[1]{\begin{pmatrix}#1\end{pmatrix}}
\newcommand{\smx}[1]{\bigl(\begin{smallmatrix}#1\end{smallmatrix}\bigr)}

\newcommand{\C}{\mathbb{C}}
\newcommand{\R}{\mathbb{R}}

\newcommand{\Z}{\mathbb{Z}}

\newcommand{\mc}[1]{\mathcal{#1}}

\newcommand{\D}[1]{\mathcal{D}(#1)} 
\newcommand{\U}[1]{\mathrm{U}(#1)}
\newcommand{\SU}[1]{\mathrm{SU}(#1)}
\renewcommand{\S}{\mathrm{S}}
\newcommand{\e}{e}

\newcommand{\rep}{\tau}
\newcommand{\z}{z}

\newcommand{\bs}{\square}
\newcommand{\bp}{\boxplus}
\newcommand{\bm}{\boxminus}

\newtheorem{theorem}{Theorem}
\newtheorem{lemma}[theorem]{Lemma}
\newtheorem{proposition}[theorem]{Proposition}

\newtheorem{corollary}[theorem]{Corollary}

\newtheorem*{conjecture*}{Conjecture}

\theoremstyle{definition}
\newtheorem*{remark}{Remark}
\newtheorem*{example}{Example}

\newcommand{\Gr}[2][-0.1cm]{
\begin{tikzpicture}[semithick, baseline = #1,
  box/.style = {draw, fill = white, inner sep = 0, minimum height = 9, minimum width = 9},
  tri/.style = {box, inner xsep = 1.5, isosceles triangle}]
  \newcommand{\X}{0.50}
  \newcommand{\Y}{0.45}
  #2
\end{tikzpicture}}
\newcommand{\gr}[1]{\Gr{\csuse{#1}}}

\newcommand{\Place}[3]{
\begin{scope}[xshift = #2cm, yshift = #3cm]
  #1
\end{scope}}
\newcommand{\place}[3]{\Place{\csuse{#1}}{#2}{#3}}

\newcommand{\round}{.. controls +(0.3,0) and +(-0.3,0) ..}

\csdef{S1}{
  \draw (0, \Y)   --   (\X, \Y);
  \draw (0,  0)   --   (\X,  0);
  \draw (0,-\Y)   --   (\X,-\Y);
}
\csdef{S2}{
  \draw (0, \Y) \round (\X,  0);
  \draw (0,  0) \round (\X,-\Y);
  \draw (0,-\Y) \round (\X, \Y);
}
\csdef{S3}{
  \draw (0, \Y) \round (\X,-\Y);
  \draw (0,  0) \round (\X, \Y);
  \draw (0,-\Y) \round (\X,  0);
}
\csdef{S4}{
  \draw (0, \Y)   --   (\X, \Y);
  \draw (0,  0) \round (\X,-\Y);
  \draw (0,-\Y) \round (\X,  0);
}
\csdef{S5}{
  \draw (0, \Y) \round (\X,  0);
  \draw (0,  0) \round (\X, \Y);
  \draw (0,-\Y)   --   (\X,-\Y);
}
\csdef{S6}{
  \draw (0, \Y) \round (\X,-\Y);
  \draw (0,  0)   --   (\X,  0);
  \draw (0,-\Y) \round (\X, \Y);
}

\csdef{T1}{\csuse{S1}}
\csdef{T2}{\csuse{S3}} 
\csdef{T3}{\csuse{S2}} 
\csdef{T4}{\csuse{S4}}
\csdef{T5}{\csuse{S5}}
\csdef{T6}{\csuse{S6}}

\csdef{Bx}#1#2#3{
  \place{S1}{1*\X}{0}
  \node[box] at (1.5*\X, \Y) {\scriptsize #1};
  \node[box] at (1.5*\X,  0) {\scriptsize #2};
  \node[box] at (1.5*\X,-\Y) {\scriptsize #3};
}
\csdef{Bxs}{\Bx{1}{2}{3}}

\csdef{Tri}#1#2#3{
  \place{S1}{1*\X}{0}
  \node[tri] at (1.5*\X, \Y) {\scriptsize #1};
  \node[tri] at (1.5*\X,  0) {\scriptsize #2};
  \node[tri] at (1.5*\X,-\Y) {\scriptsize #3};
}

\newcommand{\db}{0.12} 
\newcommand{\loopr}{0.15} 
\newcommand{\loopR}{0.35} 

\newcommand{\trloop}[4]{
  \coordinate (I3) at (#1*\X,#3*\Y);  \coordinate (L3) at ($(I3)+(0,-0.4*\Y-\db)$);
  \coordinate (O3) at (#2*\X,#3*\Y);  \coordinate (R3) at ($(O3)+(0,-0.4*\Y-\db)$);
  \draw (I3) ..controls +(-\loopr,0) and +(-\loopr,0) .. (L3)
     -- (R3) ..controls +( \loopr,0) and +( \loopr,0) .. (O3)#4; 
}

\newcommand{\Trloop}[3]{
  \coordinate (I2) at (#1*\X,#3*\Y);  \coordinate (L2) at ($(I2)+(0,-1.4*\Y-2*\db)$);
  \coordinate (O2) at (#2*\X,#3*\Y);  \coordinate (R2) at ($(O2)+(0,-1.4*\Y-2*\db)$);
  \draw (I2) ..controls +(-\loopR,0) and +(-\loopR,0) .. (L2)
     -- (R2) ..controls +( \loopR,0) and +( \loopR,0) .. (O2);
}

\newcommand{\Cl}[2]{
  \place{S#1}{0*\X}{1*\Y}
  \place{Bxs}{0*\X}{1*\Y}
  \place{T#2}{2*\X}{1*\Y}
  \draw (-0.4*\X,2*\Y) -- (     0,2*\Y);
  \draw (   3*\X,2*\Y) -- (3.4*\X,2*\Y);
  \trloop{0}{3}{0}{}
  \Trloop{0}{3}{1}
}

\newcommand{\IOO}[3]{
  \draw (-0.2*\X,0) -- (3.2*\X,0);
  \trloop{0}{1}{-1}{--cycle}
  \trloop{2}{3}{-1}{--cycle}
  \node[box] at (1.5*\X,  0) {\scriptsize #1};
  \node[box] at (0.5*\X,-\Y) {\scriptsize #2};
  \node[box] at (2.5*\X,-\Y) {\scriptsize #3};
}
\csdef{C11}{\IOO{1}{2}{3}}
\csdef{C22}{\IOO{2}{1}{3}}
\csdef{C33}{\IOO{3}{1}{2}}
\csdef{C44}{\IOO{1}{2}{3}}
\csdef{C55}{\IOO{2}{1}{3}}
\csdef{C66}{\IOO{3}{1}{2}}

\newcommand{\IIO}[3]{
  \draw (-0.2*\X,0) -- (3.2*\X,0);
  \trloop{1}{2}{-1}{--cycle}
  \node[box] at (  1*\X,  0) {\scriptsize #1};
  \node[box] at (  2*\X,  0) {\scriptsize #2};
  \node[box] at (1.5*\X,-\Y) {\scriptsize #3};
}
\csdef{C51}{\IIO{2}{1}{3}}
\csdef{C61}{\IIO{3}{1}{2}}
\csdef{C42}{\IIO{1}{2}{3}}
\csdef{C43}{\IIO{1}{3}{2}}
\csdef{C53}{\IIO{2}{3}{1}}
\csdef{C62}{\IIO{3}{2}{1}}

\newcommand{\OOO}[3]{
  \draw (-0.2*\X,0) -- (3.2*\X,0);
  \trloop{0.5}{2.5}{-1}{--cycle}
  \node[box] at (1.5*\X,  0) {\scriptsize #1};
  \node[box] at (  1*\X,-\Y) {\scriptsize #2};
  \node[box] at (  2*\X,-\Y) {\scriptsize #3};
}
\csdef{C41}{\OOO{1}{2}{3}}
\csdef{C52}{\OOO{2}{3}{1}}
\csdef{C63}{\OOO{3}{1}{2}}

\newcommand{\III}[3]{
  \draw (-0.2*\X,-0.5*\Y) -- (3.2*\X,-0.5*\Y);
  \node[box] at (0.5*\X,-0.5*\Y) {\scriptsize #1};
  \node[box] at (1.5*\X,-0.5*\Y) {\scriptsize #2};
  \node[box] at (2.5*\X,-0.5*\Y) {\scriptsize #3};
}
\csdef{C21}{\III{2}{3}{1}}
\csdef{C31}{\III{3}{2}{1}}
\csdef{C32}{\III{3}{1}{2}}
\csdef{C54}{\III{2}{3}{1}}
\csdef{C64}{\III{3}{2}{1}}
\csdef{C65}{\III{3}{1}{2}}

\newcommand{\rod}[4]{
  \tikzmath{
    \dx = #3-#1;
    \dy = #4-#2;
    \L = sqrt((\dx)^2+(\dy)^2);
    coordinate \vx, \vy;
    \w = 0.08;
    \vx = \w*( \dx/\L, \dy/\L);
    \vy = \w*(-\dy/\L, \dx/\L);
    coordinate \A, \B;
    \A = (#1,#2);
    \B = (#3,#4);
    coordinate \W;
    \W1 = (\A) - (\vx) - (\vy);
    \W2 = (\A) - (\vx) + (\vy);
    \W3 = (\B) + (\vx) + (\vy);
    \W4 = (\B) + (\vx) - (\vy);
  }
  \draw[fill = white, rounded corners = 3pt] (\W1) -- (\W2) -- (\W3) -- (\W4) -- cycle;
  \begin{scope}[pt/.style = {circle, fill = black, inner sep = 1.0pt}]
    \node[pt] at (\A) {};
    \node[pt] at (\B) {};
  \end{scope}
}

\newcommand{\linkage}{
\begin{tikzpicture}[> = latex, scale = 1.2]

  \def\z{3}

  \def\xx{\z*1.3}
  \def\yy{\z*0.7}

  \node at (\xx+0.3,0) {$\Re$};
  \node at (0,\yy+0.3) {$\Im$};
  \draw[->] (-0.2*\z,0) -- (\xx,0);
  \draw[->] (0,-0.2*\z) -- (0,\yy);
  \node at (-0.2,-0.3) {0};
  \node at (  \z,-0.3) {1};

  \def\ax{\z*0.126156}
  \def\ay{\z*0.388267}
  \def\bx{\z*0.655937}
  \def\by{\z*0.617755}

  \rod{0}{0}{\ax}{\ay}
  \rod{\bx}{\by}{\z}{0}
  \rod{\ax}{\ay}{\bx}{\by}

  \path (\ax/2,\ay/2) + (-0.30,0.15) node {$q_1$};
  \path (\ax/2+\bx/2,\ay/2+\by/2) + (-0.05,0.25) node {$q_2$};
  \path (\bx/2+\z/2,\by/2) + (0.2,0.15) node {$q_3$};

\end{tikzpicture}
}

\newcommand{\clock}{
\begin{tikzpicture}[scale = 1.2]
  \foreach \f / \t / \ax / \ay / \bx / \by in
    {3 \bm (1 \bp 2) /  0 / 0.1691020 /  0.9855986 / 1.1547005 /  0.8164966,
     1 \bp (2 \bp 3) /  1 / 0.5773503 /  0.8164966 / 1.5629488 /  0.9855986,
     2 \bp (3 \bm 1) /  2 / 0.9855986 /  0.1691020 / 1.5629488 /  0.9855986,
     3 \bm (1 \bm 2) /  3 / 0.9855986 / -0.1691020 / 1.1547005 /  0.8164966,
     1 \bm (2 \bp 3) /  4 / 0.5773503 / -0.8164966 / 0.7464522 /  0.1691020,
     2 \bp (3 \bp 1) /  5 / 0.1691020 / -0.9855986 / 0.7464522 / -0.1691020,
     3 \bp (1 \bm 2) /  6 / 0.1691020 / -0.9855986 / 1.1547005 / -0.8164966,
     1 \bm (2 \bm 3) /  7 / 0.5773503 / -0.8164966 / 1.5629488 / -0.9855986,
     2 \bm (3 \bp 1) /  8 / 0.9855986 / -0.1691020 / 1.5629488 / -0.9855986,
     3 \bp (1 \bp 2) /  9 / 0.9855986 /  0.1691020 / 1.1547005 / -0.8164966,
     1 \bp (2 \bm 3) / 10 / 0.5773503 /  0.8164966 / 0.7464522 / -0.1691020,
     2 \bm (3 \bm 1) / 11 / 0.1691020 /  0.9855986 / 0.7464522 /  0.1691020 }
    {
      \tikzmath{
        \xs = 6*cos(180+60+15+360*\t/12);
        \ys = 6*sin(180+60+15+360*\t/12);
      }
      \def\mx{1.7320508}
      \node at (\mx/2 + 0.6*\xs,0.6*\ys) {$\f$};
      \begin{scope}[xshift = \xs cm, yshift = \ys cm]
        \draw (0,0) -- (\mx,0);
        \rod{0}{0}{\ax}{\ay}
        \rod{\bx}{\by}{\mx}{0}
        \rod{\ax}{\ay}{\bx}{\by}
      \end{scope}
    }
\end{tikzpicture}
}

\begin{document}

\title{How to combine three quantum states}
\date{February 27, 2017}
\author{Maris Ozols}
\affiliation{Department of Applied Mathematics and Theoretical Physics,
University of Cambridge, Cambridge, CB3 0WA, U.K.}
\email{marozols@gmail.com}
\homepage{http://marozols.com}

\maketitle

\begin{abstract}
We devise a ternary operation for combining three quantum states: it consists of permuting the input systems in a continuous fashion and then discarding all but one of them.  This generalizes a binary operation recently studied by Audenaert et al.~\cite{ADO16}
in the context of entropy power inequalities.  Our ternary operation continuously interpolates between all such nested binary operations.  Our construction is based on a unitary version of Cayley's theorem: using representation theory we show that any finite group can be naturally embedded into a continuous subgroup of the unitary group.  Formally, this amounts to characterizing when a linear combination of certain permutations is unitary.
\end{abstract}


\section{Introduction}

A basic result in group theory known as \emph{Cayley's theorem} states that every finite group $G$ is isomorphic to a subgroup of the symmetric group.  In other words, the elements of $G$ can be faithfully represented by permutation matrices of size $\abs{G} \times \abs{G}$.  This gives a natural embedding of $G$ in the symmetric group $\S_n$ on $n = \abs{G}$ elements, known as the \emph{regular representation} of $G$.

Since permutation matrices are unitary, one might ask whether the resulting subgroup of $\S_{\abs{G}}$ can be further extended to a \emph{continuous} subgroup of the unitary group $\U{\abs{G}}$.  Such extension would allow to treat the otherwise discrete group $G$ as continuous (e.g., it would allow to perturb the elements of $G$ and continuously interpolate between them in a meaningful sense).

To illustrate this, consider the simplest non-trivial example, namely, the finite group $\Z_2$.  This group has two elements which, according to Cayley's theorem, can be represented by $2 \times 2$ permutation matrices as $I := \smx{1&0\\0&1}$ and $X := \smx{0&1\\1&0}$.  Interestingly, the complex linear combination
\begin{equation}
  U(\varphi, \alpha)
 := e^{i \varphi} \of[\big]{ \cos \alpha \, I + i \sin \alpha \, X }
  = e^{i \varphi} \mx{\cos \alpha & i \sin \alpha \\ i \sin \alpha & \cos \alpha}
  \label{eq:Ufa}
\end{equation}
of $I$ and $X$ is unitary for any $\varphi,\alpha \in [0, 2\pi)$.  Note that $U(0,0) = I$ and $U(3\pi/2,\pi/2) = X$, so by changing $\varphi$ and $\alpha$ we can continuously interpolate between the two original matrices $I$ and $X$ that represent $\Z_2$.  In fact, $U(\varphi, \alpha) U(\varphi', \alpha') = U(\varphi + \varphi', \alpha + \alpha')$, so the matrices $U(\varphi, \alpha)$ themselves form a group---a continuous two-parameter subgroup of $\U{2}$.

\subsection{Summary and main contributions}

One of our main mathematical contributions is a generalization of the above idea to \emph{any} finite group $G$: using representation theory, we characterize when a complex linear combination of matrices from the regular representation of $G$ is unitary (see \cref{thm:Unitary L}).  Equivalently, this characterizes the intersection of $\U{\abs{G}}$ and $\C[G]$, where $\C[G]$---also known as the \emph{group algebra} of $G$---consists of all formal complex linear combinations of elements of $G$.  As a consequence, any finite group $G$ can be naturally embedded into a continuous subgroup of the unitary group $\U{\abs{G}}$.  This subgroup is isomorphic to a direct sum of smaller unitary groups---one copy of $\U{d_{\rep}}$ for each irreducible representation $\rep$ of $G$, where $d_\rep$ is the dimension of $\rep$ (see \cref{cor:Cayley}).

We apply this result in quantum information theory by introducing a ternary operation that combines three quantum states.  Our operation generalizes a similar binary operation studied recently in the context of entropy power inequalities~\cite{ADO16} and quantum algorithms~\cite{LMR14,LMR+}.  This binary operation can be regarded as a qudit analogue of a beam splitter with two input and two output ports (with one of the outputs immediately discarded).  Similarly, our ternary operation is an analogue of a three-input and three-output beam splitter (with two of the outputs discarded).

More formally, the general operation consists of applying a joint unitary transformation $U$ on all three input states, followed by discarding all output systems except the first.  We require the unitary $U$ to be a complex linear combination of three-qudit permutations, thus guaranteeing that the overall operation is \emph{covariant}, i.e., commutes with any local basis change that is applied simultaneously to all three systems.  To characterize such unitaries $U$, we use the above group-theoretic result with $G = \S_3$ where $\S_3$ denotes the \emph{symmetric group} on three elements (see \cref{cor:Unitary}).

By employing graphical tensor network notation, we derive a general expression for the output state in \cref{sect:Contraction}, and then explicitly parametrize it in \cref{sect:General3} using the irreducible representations of $\S_3$.  In \cref{sect:Parametrizations}, we derive alternative parametrizations that provide further insight into our operation.  In particular, in \cref{sect:Linkage} we relate the parameters to those of a \emph{four-bar linkage} mechanism.

Using these insights, we show how nested binary operations (i.e., ones that combine the input states, two at a time) are special cases of our ternary operation (see \cref{lem:Nested}).  We also characterize the number of output state orbits when the output weights are fixed (there can be either $1$ or $2$ continuous one-parameter orbits, see \cref{prop:Orbits}).  Finally, in \cref{sect:Uniform} we investigate the \emph{uniform combination} where all states appear with equal weights in the output.

\subsection{The partial swap operation}

Ability to make discrete groups continuous has an interesting application to quantum information.  If we apply our construction to the symmetric group $\S_n$, we get a unitary extension of $\S_n$ that continuously interpolates between different permutations.  This then can be further extended to permuting $n$ quantum systems in a continuous fashion.

To illustrate this, let us focus again on the simplest non-trivial instance, $\S_2$.  This group can act on two $d$-dimensional quantum systems (\emph{qudits}) either by leaving them alone or by swapping them.  In other words, we can represent the elements of $\S_2$ by $d^2 \times d^2$ permutation matrices $I$ and $S$, where $I$ is the identity matrix and $S \ket{i}\ket{j} := \ket{j}\ket{i}$, for all $i,j \in \set{1, \dotsc, d}$, is the \emph{swap} operation.  Making this representation continuous would allow to interpolate between these two operations and thus swap two quantum systems in a continuous fashion.  This idea was exactly the starting point of~\cite{ADO16} (see also \cite{LMR14,LMR+} where it has been explored in the context of quantum algorithms).

Following \cite{ADO16}, we define the \emph{partial swap} operation $U_\lambda \in \U{d^2}$ for $\lambda \in [0,1]$ as the following linear combination of the identity $I$ and the two-qudit swap $S$:
\begin{equation}
  U_\lambda := \sqrt{\lambda} \, I + i \, \sqrt{1-\lambda} \, S
  \label{eq:Ulambda}
\end{equation}
This can easily be verified to be unitary.  In fact, $U_\lambda$ is very similar to $U(\varphi,\alpha)$ in \cref{eq:Ufa}, except we ignore the global phase $e^{i \varphi}$ and take only one sector of the unit circle corresponding to $\alpha \in [0, \pi/2]$.  Note also that $U_0 = iS$ rather than $S$; however, this global phase mismatch will be unimportant.

Let $\D{d}$ denote the set of all $d \times d$ \emph{density matrices} or \emph{qudit states}.  Given two states $\rho_1, \rho_2 \in \D{d}$, we can \emph{combine} them using the partial swap as follows:
\begin{equation}
  \rho_1 \bp_\lambda \rho_2
  := \Tr_2 \sof[\Big]{U_\lambda (\rho_1 \x \rho_2) U_\lambda\ct}
  \label{eq:box2}
\end{equation}
where $\lambda \in [0,1]$ and $\Tr_2$ denotes the \emph{partial trace} over the second system.  With some tricks (see \cref{sect:Contraction}), \cref{eq:box2} can be expanded as
\begin{equation}
  \rho_1 \bp_\lambda \rho_2
   := \lambda \rho_1 + (1-\lambda) \rho_2 + \sqrt{\lambda(1-\lambda)} \, i[\rho_2,\rho_1]
  \label{eq:box}
\end{equation}
where $[A,B] := AB-BA$ denotes the \emph{commutator} of $A$ and $B$.  Note that $\rho_1 \bp_0 \rho_2 = \rho_2$ and $\rho_1 \bp_1 \rho_2 = \rho_1$.  Moreover, if the two states commute, i.e., $[\rho_2,\rho_1] = 0$, then $\rho_1 \bp_\lambda \rho_2 = \lambda \rho_1 + (1-\lambda) \rho_2$ is just a convex combination of the two states.

The operation for combining two states obeys some interesting properties.  For example, it is \emph{covariant} under any unitary change of basis.  That is, for any $V \in \U{d}$,
\begin{equation}
  (V \rho_1 V\ct) \bp_\lambda (V \rho_2 V\ct)
  = V (\rho_1 \bp_\lambda \rho_2) V\ct.
  \label{eq:Covariance}
\end{equation}
Another, less obvious property is a quantum analogue of the \emph{entropy power inequality}~\cite{ADO16}.  To state it in full generality, consider a function $f : \D{d} \to \R$.  We say that $f$ is \emph{concave} if
\begin{equation}
  f\of[\big]{\lambda \rho + (1-\lambda) \sigma}
  \geq \lambda f(\rho) + (1-\lambda) f(\sigma)
\end{equation}
for any $\lambda \in [0,1]$ and $\rho, \sigma \in \D{d}$, and \emph{symmetric} if $f(\rho)$ depends only on the eigenvalues of $\rho$ and is symmetric in them (a simple example of a function $f$ satisfying these properties is the \emph{von~Neumann entropy}).  The following result was originally obtained in~\cite{ADO16} (see also~\cite{CLL16} for a very elegant alternative proof).

\begin{theorem}[\cite{ADO16}]\label{thm:EPI}
Let $d \geq 2$, $\rho, \sigma \in \D{d}$, $\lambda \in [0,1]$, and $\rho \boxplus_\lambda \sigma$ be as in \cref{eq:box}. If $f$ is concave and symmetric then
\begin{equation}
  f(\rho \boxplus_\lambda \sigma) \geq \lambda f(\rho) + (1-\lambda) f(\sigma).
  \label{eq:EPI}
\end{equation}
\end{theorem}

Motivated by this result, our goal is to generalize the partial swap operation to \emph{any} number of systems.  The rest of this paper is devoted to obtaining such generalization and using it to derive an analogue of \cref{eq:box} for three states.  Proving an analogue of \cref{thm:EPI} for this generalization is left as an open problem.

\subsection{The most general way of combining two states} \label{sect:General U}

Before we attempt to generalize \cref{eq:box}, it is worthwhile first checking whether $U_\lambda$ in \cref{eq:Ulambda} is actually the most general unitary matrix that can be expressed as a complex linear combination of $I$ and $S$.  As it turns out, up to an overall global phase and the sign of $i$, this is indeed the case.

\begin{proposition}\label{prop:General U}
Let $\z_1, \z_2 \in \C$ and $U := \z_1 I + \z_2 S$ where $S$ swaps two qudits ($d \geq 2$).  Then $U$ is unitary if and only if
$\z_1 =       e^{i \varphi} \sqrt{\lambda}$ and
$\z_2 = \pm i e^{i \varphi} \sqrt{1-\lambda}$
for some $\lambda \in [0,1]$ and $\varphi \in [0,2\pi)$.
In other words,
\begin{equation}
  U = e^{i \varphi} \of[\big]{ \sqrt{\lambda} \, I \pm i \, \sqrt{1-\lambda} \, S}.
  \label{eq:U2}
\end{equation}
\end{proposition}

\begin{proof}
If $U := \z_1 I + \z_2 S$ for some $\z_1, \z_2 \in \C$ then
\begin{equation}
  U U\ct = (\z_1 \bar{\z}_1 + \z_2 \bar{\z}_2) I
         + (\z_1 \bar{\z}_2 + \bar{\z}_1 \z_2) S.
\end{equation}
Since $I$ and $S$ are linearly independent and we want $U U\ct = I$, we equate the two coefficients to $1$ and $0$, respectively.  This gives us
\begin{align}
  \abs{\z_1}^2 + \abs{\z_2}^2 = 1 \qquad \text{and} \qquad
  \z_1 \bar{\z}_2 = - \overline{\z_1 \bar{\z}_2}.
  \label{eq:Two constraints}
\end{align}
From the first equation, $\z_1 = e^{i \varphi_1} \sqrt{\lambda}$ and $\z_2 = e^{i \varphi_2} \sqrt{1-\lambda}$ for some $\lambda \in [0,1]$ and $\varphi_1, \varphi_2 \in [0,2\pi)$.  The second equation says that $\z_1 \bar{\z}_2$ is purely imaginary.  Thus $e^{i (\varphi_1 - \varphi_2)} = \pm i$, which is equivalent to $e^{i \varphi_2} = \pm i e^{i \varphi_1}$.  In other words, we can take $e^{i \varphi_1} := e^{i \varphi}$ and $e^{i \varphi_2} := \pm i e^{i \varphi}$ for some $\varphi \in [0,2\pi)$.  The reverse direction follows trivially by plugging in the values of $z_1$ and $z_2$ in \cref{eq:Two constraints}.
\end{proof}

Since the map in \cref{eq:box2} involves conjugation by $U$, the global phase $e^{i \varphi}$ in \cref{eq:U2} is irrelevant and we only have the freedom of choosing the sign in front of $i$.  Up to the sign, \cref{eq:Ulambda} indeed provides the most general unitary for our purpose.  To account for the two possible signs, we define
\begin{align}
  \rho_1 \bp_\lambda \rho_2
  &:= \lambda \rho_1 + (1-\lambda) \rho_2 + \sqrt{\lambda(1-\lambda)} \, i[\rho_2,\rho_1], \label{eq:+}\\
  \rho_1 \bm_\lambda \rho_2
  &:= \lambda \rho_1 + (1-\lambda) \rho_2 - \sqrt{\lambda(1-\lambda)} \, i[\rho_2,\rho_1]. \label{eq:-}
\end{align}
These two operations are related as follows:
\begin{equation}
  \rho_1 \bp_\lambda \rho_2 = \rho_2 \bm_{1-\lambda} \rho_1.
  \label{eq:pm}
\end{equation}
Alternatively, one can introduce a parameter $t \in [0,\pi)$ such that $\abs{\cos t} = \sqrt{\lambda}$ and $\abs{\sin t} = \sqrt{1-\lambda}$.  Then the two branches given by \cref{eq:+,eq:-} can be combined into a single expression:
\begin{equation}
  \of{\cos t}^2 \rho_1 + \of{\sin t}^2 \rho_2 + \cos t \sin t \, i[\rho_2,\rho_1].
\end{equation}

\section{Generalization to three states}

Now that we fully understand the case of two states, let us investigate how to combine three states.  While it is not immediately obvious how to proceed, one simple idea is to combine them in an iterative manner by \emph{nesting} the original operation.  This turns out to preserve some of the nice properties of the original operation.  But it is not fully satisfying, since different nesting orders generally lead to different outputs.  This ambiguity motivates the introduction of a more general \emph{ternary} operation that combines all three states at once in a more symmetric fashion.  However, before we do this, it is still worthwhile to work out nesting in more detail---this will highlight some of the issues that will appear later and will also serve as a useful reference to come back to once we have the general operation.

\subsection{A nested generalization} \label{sect:Nested}

A straightforward way of generalizing \cref{eq:box} to more systems is by feeding the output state into another operation of the same kind.  For example, take $a,a' \in [0,1]$ and consider
\begin{align}
  \rho_1 \bp_a \of{\rho_2 \bp_{a'} \rho_3}
  ={}& a \rho_1 + (1-a) \of{\rho_2 \bp_{a'} \rho_3} + \sqrt{a(1-a)} \, i[\rho_2 \bp_{a'} \rho_3, \rho_1] \\
  ={}& a \rho_1 + (1-a) \of[\Big]{a' \rho_2 + (1-a') \rho_3 + \sqrt{a'(1-a')} \, i[\rho_3, \rho_2]}              \label{eq:d1} \\
    &+ \sqrt{a(1-a)} \, i\sof[\Big]{a' \rho_2 + (1-a') \rho_3 + \sqrt{a'(1-a')} \, i[\rho_3, \rho_2], \; \rho_1} \label{eq:d2} \\
  ={}& a \rho_1 + (1-a) a' \rho_2 + (1-a)(1-a') \rho_3 \label{eq:aa} \\
    &+ (1-a)\sqrt{a'(1-a')}    \, i[\rho_3, \rho_2] \label{eq:com32} \\
    &+ \sqrt{a(1-a)} \, a'     \, i[\rho_2, \rho_1] \label{eq:com21} \\
    &+ \sqrt{a(1-a)} \, (1-a') \, i[\rho_3, \rho_1] \label{eq:com31} \\
    &+ \sqrt{a(1-a)} \sqrt{a'(1-a')} \, i[i[\rho_3, \rho_2], \rho_1]. \label{eq:dcom}
\end{align}

One can easily check that this nested operation inherits the covariance property (see \cref{eq:Covariance}) of the original operation.  Moreover, it also obeys a simple generalization of inequality~\eqref{eq:EPI}.  Indeed, it follows trivially from \cref{thm:EPI} that
\begin{align}
  f \of[\big]{ \rho_1 \bp_a \of{\rho_2 \bp_{a'} \rho_3} }
  &\geq a f(\rho_1) + (1-a) f\of[\big]{ \rho_2 \bp_{a'} \rho_3 } \\
  &\geq a f(\rho_1) + (1-a) a' f(\rho_2) + (1-a)(1-a') f(\rho_3),
\end{align}
for any symmetric and concave function $f: \D{d} \to \R$.

However, there is no reason why one should single out this particular way of combining three sates.  Indeed, there are other nested combinations, such as
\begin{equation}
  \of{\rho_1 \bp_b \rho_2} \bp_{b'} \rho_3
  = b'b \rho_1 + b'(1-b) \rho_2 + (1-b') \rho_3 + \dotsb \label{eq:cc}
\end{equation}
for some $b,b' \in [0,1]$.  Here the omitted terms are similar to \cref{eq:com32,eq:com21,eq:com31,eq:dcom}, except for the double commutator which in this case is $i[\rho_3, i[\rho_2, \rho_1]]$, thus making the expression different from the one above.  Since there are several other nested ways of combining three states, it is not clear which of them, if any, should be preferred.

Considering this, it is desirable to have a \emph{ternary} operation that treats each of the input states in the same way and combines them all at once.  Note that the coefficients
\begin{align}
  \of[\big]{a, (1-a) a', (1-a)(1-a')}
  \qquad \text{and} \qquad
  \of[\big]{b'b, b'(1-b), 1-b'}
\end{align}
in front of the states $\rho_1, \rho_2, \rho_3$ in \cref{eq:aa,eq:cc} are probability distributions.  It is thus natural to demand the combined state to be of the form
\begin{equation}
  p_1 \rho_1 + p_2 \rho_2 + p_3 \rho_3 + \dotsb \label{eq:ps}
\end{equation}
for some probability distribution $(p_1, p_2, p_3)$, where---unlike in \cref{eq:dcom}---all three states are treated symmetrically in the omitted higher-order terms (in particular, the third-order terms).

\subsection{Twelve nested ways of combining three states} \label{sect:12 ways}

We can use \cref{eq:+,eq:-} to work out all possible nested ways of combining three states.  Because of relation \eqref{eq:pm}, there are exactly 12 different nested combinations:
\begin{align}
 & \rho_1 \bs_a \of{\rho_2 \bs_{a'} \rho_3} \\
 & \rho_2 \bs_b \of{\rho_3 \bs_{b'} \rho_1} \\
 & \rho_3 \bs_c \of{\rho_1 \bs_{c'} \rho_2}
\end{align}
where each $\bs$ can independently be replaced by either $\bp$ or $\bm$.  In general, these all twelve expressions are different.  However, if we choose the parameters so that
\begin{align}
  (p_1, p_2, p_3)
  &= \of[\big]{a, a'(1-a), (1-a')(1-a)} \\
  &= \of[\big]{(1-b')(1-b), b, b'(1-b)} \\
  &= \of[\big]{c'(1-c), (1-c')(1-c), c},
\end{align}
for some probability distribution $(p_1, p_2, p_3)$, we can at least guarantee that all twelve expressions begin with the same convex combination $p_1 \rho_1 + p_2 \rho_2 + p_3 \rho_3 + \dotsb$.  Our goal is to obtain a more general ternary operation that continuously interpolates between these twelve expressions.

\subsection{Combining \texorpdfstring{$n$}{n} states at once}

Recall from \cref{eq:Ulambda} that $U_\lambda$ is a linear combination of two-qudit permutation matrices.  One way of generalizing this to $n$ systems is as follows:
\begin{equation}
  (\rho_1, \dotsc, \rho_n) \mapsto
  \Tr_{2,\dotsc,n} \sof[\Big]{U (\rho_1 \x \dotsb \x \rho_n) U\ct},
  \label{eq:boxn}
\end{equation}
where $\rho_1, \dotsc, \rho_n \in \D{d}$ are qudit states, $\Tr_{2,\dotsc,n}$ denotes the partial trace over all systems but the first, and the unitary matrix $U \in \U{d^n}$ is a linear combination of all $n!$ permutations that act on $n$ qudits.  The goal of this paper is to develop a better understanding of this special type of unitaries and the corresponding map resulting from \cref{eq:boxn}.  Here we particularly focus on the $n = 3$ case and leave it as an open problem to work out the details for general $n$.

\subsection{Graphical notation for evaluating partial traces} \label{sect:Contraction}

Let us work out the details of \cref{eq:boxn} for $n=3$.  Let $\rho_1, \rho_2, \rho_3 \in \D{d}$ and define
\begin{equation}
  \rho := \Tr_{2,3} \sof[\big]{U (\rho_1 \x \rho_2 \x \rho_3) U\ct}
  \label{eq:Box3}
\end{equation}
for some unitary matrix $U \in \U{d^3}$ such that
\begin{equation}
  U := \sum_{i=1}^6 \z_i Q_i,
  \label{eq:U3}
\end{equation}
where $\z_i \in \C$ and each $Q_i \in \U{d^3}$ is one of the six permutations that act on three qudits.  We can represent them using the following graphical notation (see \cite{WBC15} for more details):
\begin{align}
  Q_1 \, &:= \, \gr{S1} &
  Q_2 \, &:= \, \gr{S2} &
  Q_3 \, &:= \, \gr{S3} &
  Q_4 \, &:= \, \gr{S4} &
  Q_5 \, &:= \, \gr{S5} &
  Q_6 \, &:= \, \gr{S6}
  \label{eq:Qs}
\end{align}
One can easily recover the matrix representation of each $Q_i$ from this pictorial notation.  For example,
\begin{align}
  Q_2 \of[\big]{\ket{\psi_1} \x \ket{\psi_2} \x \ket{\psi_3}}
  = \Gr{\place{S2}{0}{0}\Tri{1}{2}{3}}
  = \Gr{\Tri{2}{3}{1}}
  = \ket{\psi_2} \x \ket{\psi_3} \x \ket{\psi_1},
  \label{eq:psi123}
\end{align}
for any $\ket{\psi_1}, \ket{\psi_2}, \ket{\psi_3} \in \C^d$, which is enough to determine all $d^3 \times d^3$ entries of the matrix $Q_2$.

The diagrams of $Q_i$ can be composed in a natural way and this operation is compatible with matrix multiplication (in fact, the diagrams as well as the matrices form a group).  For example,
\begin{equation}
  Q_4 Q_5
  = \Gr{\place{S4}{0}{0}\place{S5}{\X}{0}}
  = \gr{S2}
  = Q_2.
\end{equation}
Inverse permutations can be found by reversing the diagram:
\begin{equation}
  Q_2\ct
  = \gr{S2}\ct
  = \gr{S3}
  = Q_3.
\end{equation}
Except for $Q_2$ and $Q_3$, which are inverses of each other, all other $Q_i$'s are self-inverse, i.e., $Q_i\ct = Q_i$ when $i \notin \set{2,3}$.  \Cref{eq:psi123} extends easily to mixed states too---the subsystems of a product state are now permuted under the conjugation by $Q_i$.  For example,
\begin{align}
  Q_2 (\rho_1 \x \rho_2 \x \rho_3) Q_2\ct
  = \Gr{\place{S2}{0}{0}\place{Bxs}{0}{0}\place{S3}{2*\X}{0}}
  = \Gr{\Bx{2}{3}{1}}
  = \rho_2 \x \rho_3 \x \rho_1.
\end{align}

\newcommand{\evaluate}[3]{
  \xi_{#1#2}
 &= \z_#1 \bar{z}_#2 \Tr_{2,3} \sof*{
    \Gr{
      \place{S#1}{0*\X}{0*\Y}
      \place{Bxs}{0*\X}{0*\Y}
      \place{S#2}{2*\X}{0*\Y}
    }}
  = \z_#1 \bar{z}_#2 \Gr[0.3cm]{\Cl{#1}{#2}}
  = \z_#1 \bar{z}_#2 \, \Gr[-0.33cm]{\csuse{C#1#2}}
  = \z_#1 \bar{z}_#2 \, #3
}

Using this pictorial representation, we would like to expand $U$ on both sides of \cref{eq:Box3} and obtain an explicit formula for the output state $\rho$ in terms of the input states $\rho_k$ and the coefficients $\z_i$.  This requires computing the matrices
\begin{equation}
  \xi_{ij} := \Tr_{2,3} \sof[\big]{\z_i Q_i (\rho_1 \x \rho_2 \x \rho_3) \bar{\z}_j Q_j\ct},
\end{equation}
for each combination of $i,j \in \set{1, \dotsc, 6}$.  This daunting task becomes much more straightforward using the graphical notation for the partial trace~\cite{WBC15}.  For example,
\begin{align}
  \evaluate{1}{1}{\rho_1 \Tr \rho_2 \Tr \rho_3}, \\
  \evaluate{4}{1}{\rho_1 \Tr(\rho_2 \rho_3)},    \\
  \evaluate{5}{1}{\rho_2 \rho_1 \Tr \rho_3},     \\
  \evaluate{2}{1}{\rho_2 \rho_3 \rho_1}.
\end{align}
The remaining cases are similar and are summarized in \cref{fig:Contraction}.

We can find the resulting state $\rho = \sum_{i,j=1}^6 \xi_{ij}$, defined in \cref{eq:Box3}, by putting all 36 terms together:
\begin{falign}
  \rho
  {}={}& \of[\big]{ \abs{\z_1}^2 + \abs{\z_4}^2 + 2 \Re(\z_1 \bar{\z}_4) \Tr(\rho_2 \rho_3) } \rho_1 \nonumber \\
  {}+{}& \of[\big]{ \abs{\z_2}^2 + \abs{\z_5}^2 + 2 \Re(\z_2 \bar{\z}_5) \Tr(\rho_3 \rho_1) } \rho_2 \nonumber \\
  {}+{}& \of[\big]{ \abs{\z_3}^2 + \abs{\z_6}^2 + 2 \Re(\z_3 \bar{\z}_6) \Tr(\rho_1 \rho_2) } \rho_3 \nonumber \\[0.2cm]
  {}+{}& \of[\big]{ \z_1 \bar{\z}_5 + \z_4 \bar{\z}_2 } \rho_1 \rho_2 + c.t. \nonumber \\
  {}+{}& \of[\big]{ \z_2 \bar{\z}_6 + \z_5 \bar{\z}_3 } \rho_2 \rho_3 + c.t. \label{eq:Magic} \\
  {}+{}& \of[\big]{ \z_3 \bar{\z}_4 + \z_6 \bar{\z}_1 } \rho_3 \rho_1 + c.t. \nonumber \\[0.2cm]
  {}+{}& \of[\big]{ \z_2 \bar{\z}_1 + \z_5 \bar{\z}_4 } \rho_2 \rho_3 \rho_1 + c.t. \nonumber \\
  {}+{}& \of[\big]{ \z_3 \bar{\z}_2 + \z_6 \bar{\z}_5 } \rho_3 \rho_1 \rho_2 + c.t. \nonumber \\
  {}+{}& \of[\big]{ \z_1 \bar{\z}_3 + \z_4 \bar{\z}_6 } \rho_1 \rho_2 \rho_3 + c.t. \nonumber
\end{falign}
where $c.t.$ stands for the conjugate transpose of the previous term.

\begin{figure}
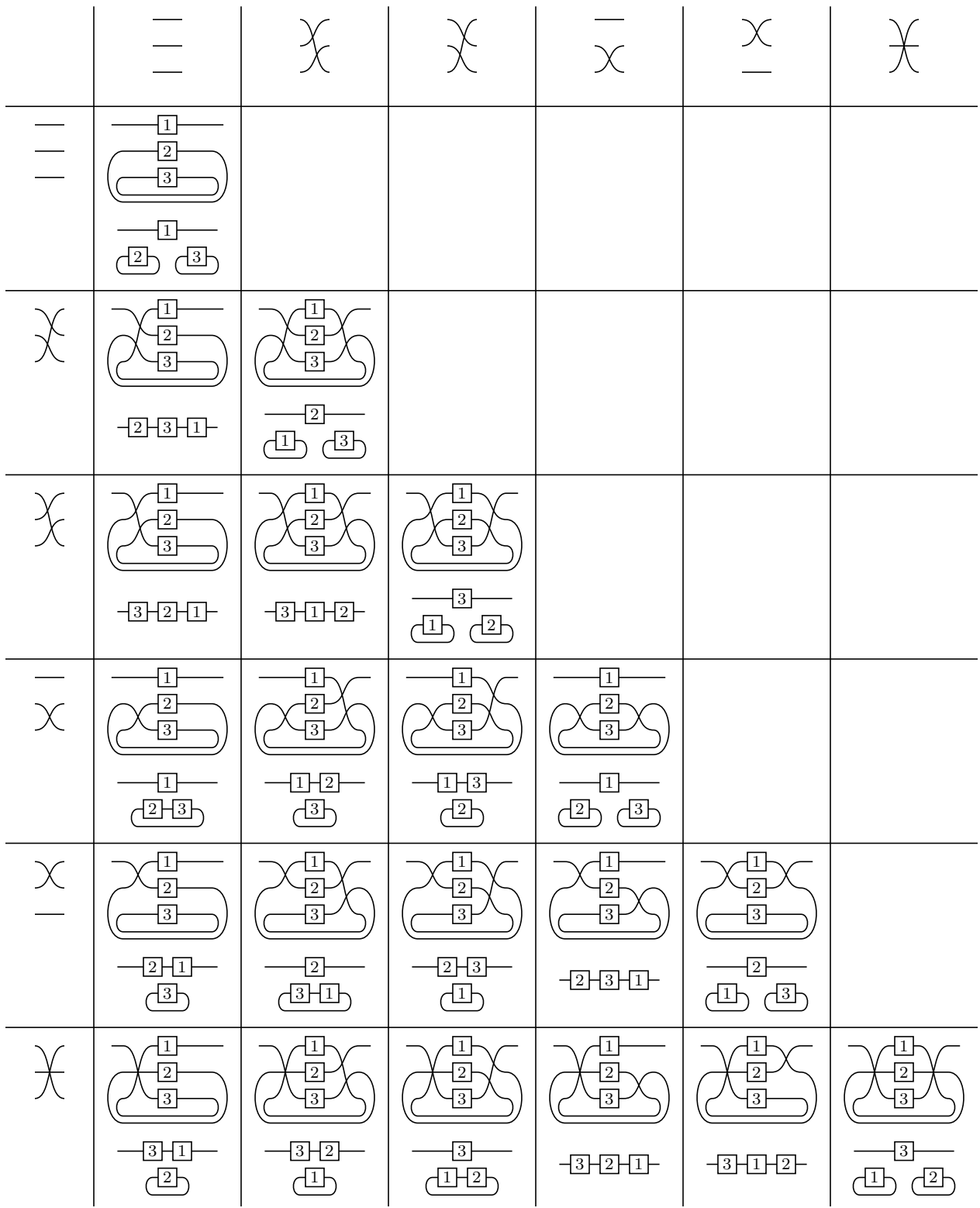

\centering
\Gr{
  \def\x{5*\X}
  \def\y{-7*\Y}
  \foreach \i in {1,...,6} {
    \draw (\X,\y*\i+2.7*\Y) -- (\x*7-\X,\y*\i+2.7*\Y);
    \draw (\x*\i-\X,-0.5*\Y) -- (\x*\i-\X,\y*7+2.9*\Y);
    \place{T\i}{\X cm+\x*\i}{-2*\Y}
    \place{S\i}{2*\X}{\Y cm+\i*\y}
    \foreach \j in {\i,...,6} {
      \Place{\Cl{\j}{\i}}{\x*\i}{\y*\j}
      \pgfmathparse{\x*\i}     \let\xx=\pgfmathresult
      \pgfmathparse{\y*\j-2*\Y}\let\yy=\pgfmathresult
      \place{C\j\i}{\xx}{\yy}
    }
  }
}
\caption{\label{fig:Contraction}Tensor contraction diagrams for computing $\xi_{ij} := \Tr_{2,3} \sof[\big]{Q_i (\rho_1 \x \rho_2 \x \rho_3) Q_j\ct}$ for each pair of three qudit permutations $Q_i$ (rows) and $Q_j\ct$ (columns), for $i \geq j$.  By combining all 36 terms we obtain \cref{eq:Magic}.}
\end{figure}

\subsection{Unitarity and independence constraints}

Before we can call \cref{eq:Magic} a generalization of \cref{eq:box}, we need to answer the following two questions:

\begin{enumerate}[Q1.]
  \item \label{Q1} \textbf{Unitarity:}
  Under what constraints on the coefficients $\z_i \in \C$ is the matrix $U$ in \cref{eq:U3} unitary?
  \item \label{Q2} \textbf{Independence:}
  Are additional constraints needed to make the $\Tr (\rho_i \rho_j)$ terms in \cref{eq:Magic} vanish?
\end{enumerate}

\newcommand{\Q}[1]{\hyperref[Q#1]{Q#1}}

We need to demand the unitarity of $U$ since we want $\rho$ to be a valid quantum state---this is very natural considering \cref{eq:Box3}.  Note that \cref{prop:General U} already answers \Q1 for $n = 2$.  We will answer \Q1 in full generality using representation theory: \cref{cor:Unitary} characterizes, for any $n \geq 1$, when a linear combination $U := \sum_{\pi \in \S_n} \z_\pi Q_\pi$ is unitary.

The reason for asking \Q2 is because we would like to control the ``amount'' of each $\rho_i$ in the output state.  This is important in the context of inequality~\eqref{eq:EPI}, which borrows the coefficients $\lambda$ and $1-\lambda$ directly from \cref{eq:box}.  Inspired by this, we would also like the coefficients in \cref{eq:Magic}---especially, those of the first-order terms $\rho_i$---to be independent of the states themselves.  When $n=2$, this happens automatically, see \cref{eq:box}, and leads to the particularly simple form of inequality~\eqref{eq:EPI}.  For $n=3$, however, \cref{eq:Magic} has an additional $\Tr (\rho_i \rho_j)$ term within each coefficient of $\rho_i$.  As we will see, these terms survive the unitarity constraint from \Q1, so we can remove them only by demanding directly that
\begin{equation}
  \Re(\z_1 \bar{\z}_4) = \Re(\z_2 \bar{\z}_5) = \Re(\z_3 \bar{\z}_6) = 0.
  \label{eq:Res}
\end{equation}
This turns the first-order terms of \cref{eq:Magic} into a linear (in fact, a convex) combination of $\rho_i$, just like in \cref{eq:box} for $n=2$.  While \cref{eq:Res} might seem a bit arbitrary, it can be imposed in a fairly natural way (see \cref{sect:Nice}) and leads to some further nice structure in \cref{eq:Magic} (see \cref{sect:Double commutators}).

\section{When is a linear combination of permutations unitary?}

We begin by first answering \Q1, which will require some representation theory.  Particularly relevant to us is reference~\cite{Childs} that provides a very concise introduction to representation theory and Fourier analysis of non-abelian groups.  For more background on representation theory of finite groups, see the standard reference~\cite{Serre}.

\subsection{Background on representation theory}

Let $\U{d}$ denote the set of all $d \times d$ unitary matrices.  A $d$-dimensional \emph{representation} of a finite group $G$ is a map $\rep: G \to \U{d}$ such that $\rep(gh) = \rep(g) \rep(h)$ for all $g, h \in G$.  We will always use $\e$ to denote the identity element of $G$.  Note that $\rep(\e) = I_d$, the $d \times d$ identity matrix, and $\rep(g^{-1}) = \rep(g)\ct$.  We call $d$ the \emph{dimension} of the representation $\rep$.

Let $\S_n$ denote the \emph{symmetric group} consisting of all $n!$ permutations acting on $n$ elements.  We write $\pi(i) = j$ to mean that permutation $\pi \in \S_n$ maps $i$ to $j$, and we write $\pi^{-1}$ to denote the inverse permutation of $\pi$.  The following are four different representations of the symmetric group (two of them, $Q_\pi$ and $L_\pi$, will play an important role later in the paper).

\begin{example}[Representations of $\S_n$]
The symmetric group $\S_n$ can be represented by permutation matrices in several different ways. In each case we write down how the matrix associated to permutation $\pi \in \S_n$ acts in the standard basis:
\begin{itemize}
\item \emph{natural representation} $P_\pi \in \U{n}$:
\begin{equation}
  P_\pi : \ket{i} \mapsto \ket{\pi(i)},
  \qquad \forall i \in \set{1,\dotsc,n};
\end{equation}
\item \emph{tensor representation} $Q_\pi \in \U{d^n}$:
\begin{equation}
  Q_\pi : \ket{i_1} \x \dotsb \x \ket{i_n}
  \mapsto \ket{i_{\pi^{-1}(1)}} \x \dotsb \x \ket{i_{\pi^{-1}(n)}},
  \qquad \forall i_1, \dotsc, i_n \in \set{1,\dotsc,d};
  \label{eq:Qpi}
\end{equation}
\item \emph{left regular} and \emph{right regular} representations $L_\pi, R_\pi \in \U{\abs{\S_n}}$:
\begin{align}
  L_\pi &: \ket{\sigma} \mapsto \ket{\pi \sigma},      & \forall \sigma \in \S_n, \\
  R_\pi &: \ket{\sigma} \mapsto \ket{\sigma \pi^{-1}}, & \forall \sigma \in \S_n.
\end{align}
\end{itemize}
We call $d$ the \emph{local dimension} of the tensor representation (we will typically required that $d \geq n$).  Note that for the regular representations, the standard basis of the underlying space is labeled by permutations themselves, so the space has $n!$ dimensions: $\C^{\S_n} \cong \C^{n!}$.  Finally, note that the regular representations can be defined for any finite group $G$ in a similar manner.
\end{example}

Since $Q_\pi$ will play an important role later, let us verify that it is indeed a representation\footnote{It would not be a representation if we would use $\pi$ instead of $\pi^{-1}$ on the RHS of \cref{eq:Qpi}.} of $\S_n$.  First, note that
\begin{equation}
  Q_\pi \bigotimes_{i=1}^n \ket{\psi_i}
  = \bigotimes_{i=1}^n \ket{\psi_{\pi^{-1}(i)}}
  = \bigotimes_{i=1}^n \ket{\phi_i}
\end{equation}
where $\ket{\phi_i} := \ket{\psi_{\pi^{-1}(i)}}$.  Following the same rule, we see that
\begin{equation}
  Q_\sigma \bigotimes_{i=1}^n \ket{\phi_i}
  = \bigotimes_{i=1}^n \ket{\phi_{\sigma^{-1}(i)}}
  = \bigotimes_{i=1}^n \ket{\psi_{\pi^{-1}(\sigma^{-1}(i))}}
  = \bigotimes_{i=1}^n \ket{\psi_{(\sigma\pi)^{-1}(i)}}.
  \label{eq:Qsp}
\end{equation}
In other words, $Q_\sigma Q_\pi$ acts in exactly the same way as $Q_{\sigma\pi}$.

\subsection{From permutations to unitaries: the importance of the left regular representation}

Consider the following linear combination:
\begin{equation}
  U := \sum_{\pi \in \S_n} \z_\pi Q_\pi,
\end{equation}
where $\z_\pi \in \C$ and $Q_\pi$ are the matrices from the tensor representation of $\S_n$.  Can we parametrize in some simple way all coefficient tuples $\z := \of{\z_\pi : \pi \in \S_n} \in \C^{\S_n}$ such that $U$ is unitary simultaneously for all local dimensions $d \geq 1$?  We will provide an answer to this question using representation theory.

Let us start by first reducing this problem from the tensor representation $Q_\pi$ to the left regular representation $L_\pi$.  As a first step, we need to show that the matrices $L_\pi$ are linearly independent (this holds more generally, i.e., not just for $\S_n$ but for any finite group $G$).

\begin{proposition}\label{prop:Independence of L}
For any finite group $G$, the matrices $\set{L_g : g \in G}$ have disjoint supports (i.e., locations of non-zero entries) and thus are linearly independent.
\end{proposition}

\begin{proof}
For any $x,y,g \in G$,
\begin{equation}
  \bra{x} L_g \ket{y}
  = \braket{x}{g y}
  = \begin{cases}
      1 & \text{if } g = x y^{-1}, \\
      0 & \text{otherwise}.
    \end{cases}
  \label{eq:xy}
\end{equation}
So, for any fixed $x$ and $y$, there is exactly one matrix, namely $L_{x y^{-1}}$, with a non-zero entry in row $x$ and column $y$.  Hence, the matrices $L_g$ have disjoint supports and thus are linearly independent.
\end{proof}

\begin{remark}
The matrices $L_g$ form a (non-commutative) \emph{association scheme}~\cite{Ban93}.  Their linear span is known as the \emph{Bose--Mesner algebra} of this scheme.  In representation theory, it goes by the name of the \emph{group algebra} of $G$ and is denoted by $\C[G]$.  Our \cref{thm:Unitary L} (see below) characterizes all unitary matrices within this algebra.
\end{remark}

\begin{example}[$G = \S_3$]
Because of \cref{prop:Independence of L}, any linear combination of the matrices $L_g$ has a particularly nice structure.  For example, if $G = \S_3$ and we order the permutations according to \cref{eq:Qs} then
\begin{equation}
  \sum_{k=1}^6 k L_k =
  \mx{
    1 & 3 & 2 & 4 & 5 & 6 \\
    2 & 1 & 3 & 6 & 4 & 5 \\
    3 & 2 & 1 & 5 & 6 & 4 \\
    4 & 6 & 5 & 1 & 2 & 3 \\
    5 & 4 & 6 & 3 & 1 & 2 \\
    6 & 5 & 4 & 2 & 3 & 1
  }.
  \label{eq:Ls}
\end{equation}
One can easily read off the matrix representation of each $L_i$ from this.
\end{example}

Let us now show that the matrices $Q_\pi$ are also linearly independent, when $d$ is sufficiently large.

\begin{lemma}\label{Independence of Q}
If $d \geq n$ then $Q_\pi \cong L_\pi \oplus \rep(\pi)$ for some representation $\rep$ of $\S_n$, so the matrices $\set{Q_\pi : \pi \in \S_n}$ are linearly independent.
\end{lemma}

\begin{proof}
It suffices to show that $Q_\pi$ are linearly independent even when restricted to some invariant subspace of $\C^{d^n}$.  Let $\ket{\Psi_\pi} := Q_\pi \bigotimes_{i=1}^n \ket{i} = \bigotimes_{i=1}^n \ket{\pi^{-1}(i)}$ (we need $d \geq n$ for this to make sense).  Note that $Q_\sigma \ket{\Psi_\pi} = Q_\sigma \bigotimes_{i=1}^n \ket{\pi^{-1}(i)} = \bigotimes_{i=1}^n \ket{\pi^{-1}(\sigma^{-1}(i))} = \bigotimes_{i=1}^n \ket{(\sigma\pi)^{-1}(i)} = \ket{\Psi_{\sigma\pi}}$, just like in \cref{eq:Qsp}.  So $\mc{L} := \spn \set{\ket{\Psi_\pi} : \pi \in \S_n}$ is an invariant subspace of $\C^{d^n}$ under the action of the tensor representation.  The vectors $\ket{\Psi_\pi}$ have disjoint supports and hence form an orthonormal basis of the subspace $\mc{L}$.  Moreover, the tensor representation $Q_\pi$ acts as the left regular representation $L_\pi$ in this subspace.  In other words,
\begin{equation}
  Q_\pi \cong L_\pi \oplus \rep(\pi)
  \label{eq:Q and L}
\end{equation}
for some (reducible) representation $\rep$ acting on the orthogonal complement of $\mc{L}$ in $\C^{d^n}$.  According to \cref{prop:Independence of L}, the matrices $\set{L_\pi : \pi \in \S_n}$ are linearly independent, so the result follows.
\end{proof}

\begin{lemma}\label{lem:Regular rep}
Let $\z_\pi \in \C$, for each $\pi \in \S_n$.  If $d \geq n$ then $\sum_{\pi \in \S_n} \z_\pi Q_\pi$ is unitary if and only if $\sum_{\pi \in \S_n} \z_\pi L_\pi$ is unitary (the reverse implication holds for any $d \geq 1$).
\end{lemma}

\begin{proof}
As we noted in the proof of \cref{Independence of Q}, for $d \geq n$ there is a subspace $\mc{L}$ of $\C^{d^n}$ where $Q_\pi$ acts as $L_\pi$, see \cref{eq:Q and L}; this immediately gives the forward implication.  For the reverse implication, note that we can decompose any representation of a finite group into a direct sum of its irreducible representations or \emph{irreps}~\cite{Serre}.  If we obtain such decomposition for the left regular representation, a standard result from representation theory says that each irrep will appear at least once in this decomposition~\cite{Serre}.  Thus, the unitarity of $\sum_{\pi \in \S_n} \z_\pi L_\pi$ is equivalent to the simultaneous unitarity of $\sum_{\pi \in \S_n} \z_\pi \rep(\pi)$ for all irreps $\rep$ of $\S_n$.  Since the tensor representation $Q_\pi$ can also be decomposed as a direct sum of irreps, for any $d \geq 1$, we conclude that $\sum_{\pi \in \S_n} \z_\pi Q_\pi$ must therefore be unitary.
\end{proof}

According to \cref{lem:Regular rep}, our problem now reduces to analyzing the left regular representation of $\S_n$ and characterizing when $U := \sum_{\pi \in \S_n} \z_\pi L_\pi$ is unitary.  We will analyze this in the Fourier basis, where the left regular representation decomposes as a direct sum of irreducible representations.  This strategy does not rely on any special properties of $\S_n$, so we might as well work with an arbitrary finite group $G$ for the sake of generality.

\subsection{Fourier transform over finite groups}

Recall from~\cite{Childs} that the \emph{Fourier transform} over a finite group $G$ is the unitary matrix
\begin{equation}
  F := \sum_{g \in G} \sum_{\rep \in \hat{G}}
       \sqrt{\frac{d_\rep}{\abs{G}}} \sum_{j,k = 1}^{d_\rep}
       \rep(g)_{j,k} \ket{(\rep,j,k)} \bra{g},
  \label{eq:F long}
\end{equation}
where $\hat{G}$ denotes the set of irreducible representations of $G$, $d_\rep$ is the dimension of irrep $\rep$, and $\rep(g)_{j,k}$ are the matrix elements of $\rep(g)$.  The output space of $F$ is labeled by triples $(\rep,j,k)$ and has the same dimension as the input space---indeed, it is a standard result in representation theory that
\begin{equation}
  \sum_{\rep \in \hat{G}} d_\rep^2 = \abs{G}.
  \label{eq:dims}
\end{equation}
One can easily check that $F$ is unitary using the orthogonality of characters~\cite{Childs}.

We will need the following standard result from representation theory~\cite{Childs,Serre}.  Let $\hat{L}_g := F L_g F\ct$ denote $L_g$ in the Fourier basis.  Then
\begin{equation}
  \hat{L}_g = \bigoplus_{\rep \in \hat{G}} \sof[\big]{\rep(g) \x I_{d_\rep}}
  \label{eq:hatL}
\end{equation}
where $I_{d_\rep}$ is the $d_\rep \times d_\rep$ identity matrix.  In other words, $\hat{L}_g$ is block-diagonal and contains each irrep $\rep$ the number of times equal to its dimension $d_\rep$.

\subsection{Characterization of unitary linear combinations}

Let us first show a simple preliminary fact.  Recall from \cref{prop:Independence of L} that $\set{L_g : g \in G}$ is a linearly independent set.  In fact, when properly normalized, these matrices are orthonormal.

\begin{proposition}\label{prop:Basis}
For any finite group $G$, the matrices $\set{L_g / \sqrt{\abs{G}} : g \in G}$ are orthonormal with respect to the Hilbert--Schmidt inner product $\ip{A,B} := \Tr(A\ct B)$.
\end{proposition}

\begin{proof}
First, note from \cref{eq:xy} that, for any $g \in G$,
\begin{equation}
  \Tr L_g
  = \sum_{h \in G} \bra{h} L_g \ket{h}
  = \sum_{h \in G} \delta_{g,\e}
  = \abs{G} \delta_{g,\e}.
\end{equation}
Thus, for any $a,b \in G$, the corresponding Hilbert--Schmidt inner product is
\begin{equation}
  \frac{1}{\abs{G}} \ip{L_a,L_b}
  = \frac{1}{\abs{G}} \Tr \of[\big]{L_a\ct L_b}
  = \frac{1}{\abs{G}} \Tr L_{a^{-1} b}
  = \delta_{a,b}.
\end{equation}
Hence the normalized matrices $L_g / \sqrt{\abs{G}}$ are orthonormal.
\end{proof}

\begin{theorem}\label{thm:Unitary L}
Let $G$ be a finite group, $L_g$ be its left regular representation, and $\z_g \in \C$, for each $g \in G$.  Then $\sum_{g \in G} \z_g L_g$ is unitary if and only if
\begin{equation}
  \z_g
  = \sum_{\rep \in \hat{G}}
    \frac{d_\rep}{\abs{G}}
    \Tr \of[\big]{\rep(g)\ct U_\rep},
  \label{eq:zg}
\end{equation}
for some choice of $U_\rep \in \U{d_\rep}$, one for each irrep $\rep$ of $G$.
\end{theorem}

\begin{proof}
Clearly, $\sum_{g \in G} \z_g L_g$ is unitary if and only if $\sum_{g \in G} \z_g \hat{L}_g$ is unitary where
\begin{equation}
  \hat{L}_g
  = F L_g F\ct
  = \bigoplus_{\rep \in \hat{G}} \sof[\big]{\rep(g) \x I_{d_\rep}}
  \label{eq:Lg hat}
\end{equation}
according to \cref{eq:hatL}.  We prefer to work in the Fourier basis, since then all $\hat{L}_g$ become simultaneously block diagonal and we can write
\begin{equation}
  \sum_{g \in G} \z_g \hat{L}_g
  = \bigoplus_{\rep \in \hat{G}} \sof[\bigg]{
      \of[\Big]{\sum_{g \in G} \z_g \rep(g)} \x I_{d_\rep}
    }.
\end{equation}
This matrix is unitary if and only if each of its blocks is unitary, i.e.,
\begin{equation}
  \sum_{g \in G} \z_g \hat{L}_g
  = \bigoplus_{\rep \in \hat{G}} \sof[\big]{ U_\rep \x I_{d_\rep} }
  =: U
  \label{eq:U blocks}
\end{equation}
for some set of unitaries $U_\rep \in \U{d_\rep}$.  Since $\mc{B} := \set{L_g / \sqrt{\abs{G}}: g \in G}$ is an orthonormal set, see \cref{prop:Basis}, so is $\hat{\mc{B}} := \set{FBF\ct : B \in \mc{B}} = \set{\hat{L}_g / \sqrt{\abs{G}}: g \in G}$.  Since $\abs{\hat{\mc{B}}} = \abs{G} = \sum_{\rep \in \hat{G}} d_\rep^2$ according to \cref{eq:dims}, $\hat{\mc{B}}$ is in fact an orthonormal basis for the set of all block matrices that have the same block structure as $U$ in \cref{eq:U blocks}.  Thus, we can obtain the coefficients $\z_g$ in the expansion
\begin{equation}
  \frac{U}{\sqrt{\abs{G}}}
  = \sum_{g \in G} z_g \, \frac{\hat{L}_g}{\sqrt{\abs{G}}}
\end{equation}
simply by projecting on the corresponding basis vector:
\begin{equation}
  \z_g
  = \ip[\bigg]{
      \frac{\hat{L}_g}{\sqrt{\abs{G}}},
      \frac{U}{\sqrt{\abs{G}}}
    }
  = \frac{1}{\abs{G}}
    \Tr \of[\Bigg]{
      \bigoplus_{\rep \in \hat{G}}
      \sof[\Big]{ \of{\rep(g)\ct U_\rep} \x I_{d_\rep} }
    }
  = \sum_{\rep \in \hat{G}}
    \frac{d_\rep}{\abs{G}}
    \Tr \of{\rep(g)\ct U_\rep},
\end{equation}
where we substituted \cref{eq:Lg hat,eq:U blocks}.  The reverse implication follows by applying all steps in the reverse order.
\end{proof}

As a byproduct of our proof, we observe that all unitary linear combinations $\sum_{g \in G} \z_g L_g$ form a group of their own.  According to \cref{eq:U blocks}, this group is isomorphic to the direct sum of unitary groups,
\begin{equation}
  \bigoplus_{\rep \in \hat{G}} \U{d_\rep},
\end{equation}
and it contains $G$ as a subgroup (as represented by the matrices $L_g$).  Indeed, if we take any $g \in G$ and, for all irreps $\rep \in \hat{G}$, set $U_\rep := \rep(g)$ in \cref{eq:zg}, then $\z_g = 1$ while all other coefficients vanish (see \cref{prop:Basis}), reducing the linear combination to $L_g$.  Considering this, we can intuitively think of
\begin{equation}
  \set[\Big]{\sum_{g \in G} \z_g L_g \in \U{\abs{G}}: \z_g \in \C}
\end{equation}
as a natural continuous extension of the discrete finite group $G$.

\begin{corollary}[Unitary version of Cayley's theorem]\label{cor:Cayley}
Every finite group $G$ can be extended to a continuous subgroup of the unitary group $\U{\abs{G}}$.  This subgroup is isomorphic to $\bigoplus_{\rep \in \hat{G}} \U{d_\rep}$, where $\hat{G}$ is the set of all irreps of $G$ and $d_\rep$ is the dimension of irrep $\rep$.
\end{corollary}

If we specialize \cref{thm:Unitary L} to $G = \S_n$ and apply \cref{lem:Regular rep}, we get the following result.

\begin{corollary}\label{cor:Unitary}
Let $\z_\pi \in \C$, for each $\pi \in \S_n$.  If $d \geq n$ then $\sum_{\pi \in \S_n} \z_\pi Q_\pi$ is unitary if and only if
\begin{equation}
  \z_\pi
  = \sum_{\rep \in \hat{\S}_n}
    \frac{d_\rep}{n!}
    \Tr \of[\big]{\rep(\pi)\ct U_\rep},
  \label{eq:Formula}
\end{equation}
for some choice of $U_\rep \in \U{d_\rep}$, one for each irrep $\rep$ of $\S_n$ (the reverse implication holds for any $d \geq 1$).
\end{corollary}

\section{How to combine three quantum states?}

Let $Q_\pi \in \U{d^3}$ denote the matrix that permutes three qudits according to permutation $\pi \in \S_3$, see \cref{eq:Qs}.  We will now use \cref{cor:Unitary} to parametrize all tuples of complex coefficients $\of{\z_\pi : \pi \in \S_3} \in \C^6$, such that $\sum_{\pi \in \S_3} \z_\pi Q_\pi$ is a unitary matrix (for all $d \geq 1$).  First, we need to work out all irreps of $\S_3$.

\subsection{The irreducible representations of \texorpdfstring{$\S_3$}{S3}}

The symmetric group $\S_3$ has three irreducible representations~\cite{Serre}: two 1-dimensional representations (the trivial representation $\rep_1$ and the sign representation $\rep_2$) and a 2-dimensional representation $\rep_3$.  Recall that $\S_3 \cong \mathrm{D}_3$ (the dihedral group), so geometrically $\rep_3$ corresponds to rotations and reflections in 2D that preserve an equilateral triangle (centered at the origin and with one corner pointing along the $x$ axis).  These representations are written out explicitly in \cref{tab:S3 irreps}.  One can easily verify that $\rep_k(\pi \sigma) = \rep_k(\pi) \rep_k(\sigma)$ for all $\pi, \sigma \in \S_3$ and $k \in \set{1,2,3}$.  For example,
\begin{equation*}
  \rep_3 \of*{\Gr{\place{S4}{0}{0}\place{S5}{\X}{0}}}
  = \rep_3 \of*{\gr{S4}} \rep_3 \of*{\gr{S5}}
  = \mx{1&0\\0&-1} \cdot \frac{1}{2} \mx{-1 & -\sqrt{3} \\ -\sqrt{3} & 1}
  = \frac{1}{2} \mx{-1 & -\sqrt{3} \\ \sqrt{3} & -1}
  = \rep_3 \of*{\gr{S2}}.
\end{equation*}

\begin{table}
\begin{equation*}
  \newcommand{\vspc}[1]{\rule{0pt}{#1}}
  \begin{array}{c|cccccc}
           \pi  & \gr{S1} & \gr{S2} & \gr{S3} & \gr{S4} & \gr{S5} & \gr{S6} \\[10pt] \hline
    \vspc{12pt}
    \rep_1(\pi) & 1 & 1 & 1 &  1 &  1 &  1 \\
    \vspc{12pt}
    \rep_2(\pi) & 1 & 1 & 1 & -1 & -1 & -1 \\
    \rep_3(\pi) & \mx{1&0\\0&1} &
      \frac{1}{2} \mx{-1& -\sqrt{3} \\  \sqrt{3}&-1} &
      \frac{1}{2} \mx{-1&  \sqrt{3} \\ -\sqrt{3}&-1} &
                  \mx{1&0\\0&-1} &
      \frac{1}{2} \mx{-1& -\sqrt{3} \\ -\sqrt{3}& 1} &
      \frac{1}{2} \mx{-1&  \sqrt{3} \\  \sqrt{3}& 1} \\[10pt] \hline
   \vspc{20pt}
   \rep'_3(\pi) & \mx{1&0\\0&1} & \mx{\omega&0\\0&\omega^2} & \mx{\omega^2&0\\0&\omega} &
                  \mx{0&1\\1&0} & \mx{0&\omega^2\\\omega&0} & \mx{0&\omega\\\omega^2&0}
  \end{array}
\end{equation*}
\caption{\label{tab:S3 irreps}All irreducible representations of $\S_3$.  The matrix elements of $\rep_3$ are not unique but depend on the choice of basis---we provide two simple choices (irreps $\rep_3$ and $\rep'_3$ are isomorphic).}
\end{table}

\subsection{Parametrizing the general solution for \texorpdfstring{$\S_3$}{S3}} \label{sect:General3}

According to \cref{cor:Unitary}, we need to assign one unitary matrix $U_k \in \U{d_{\rep_k}}$ to each irrep $\rep_k$ of $\S_3$.  Since the global phase of $\sum_{i=1}^6 \z_i Q_i$ has no effect on the output state, we can assume without loss of generality that one of the unitaries $U_k$ is in the special unitary group.  We take $U_1, U_2 \in \U{1}$ and $U_3 \in \SU{2}$, and parametrize these unitaries as follows:
\begin{align}
  U_1 &:= \mx{e^{i \varphi_1}}, &
  U_2 &:= \mx{e^{i \varphi_2}}, &
  U_3 &:= \mx{a & c \\ -\bar{c} & \bar{a}}
  \label{eq:Us}
\end{align}
where $\varphi_1, \varphi_2 \in [0, 2\pi)$ and $a,c \in \C$ are such that $\abs{a}^2 + \abs{c}^2 = 1$.

We can now use \cref{eq:Formula} and irreps $\rep_1, \rep_2, \rep_3$ from \cref{tab:S3 irreps} to compute the coefficients $\z_1, \dotsc, \z_6$:
\begin{align}
  \z_1 &=
    \frac{1}{6} \sof[\big]{
        e^{i \varphi_1}
      + e^{i \varphi_2}
      + 4 \Re(a)
    }, \label{eq:z1} \\
  \z_2 &=
    \frac{1}{6} \sof[\big]{
        e^{i \varphi_1}
      + e^{i \varphi_2}
      - 2 \Re(a + \sqrt{3} c)
    }, \label{eq:z2} \\
  \z_3 &=
    \frac{1}{6} \sof[\big]{
        e^{i \varphi_1}
      + e^{i \varphi_2}
      - 2 \Re(a - \sqrt{3} c)
    }, \label{eq:z3} \\
  \z_4 &=
    \frac{1}{6} \sof[\big]{
        e^{i \varphi_1}
      - e^{i \varphi_2}
      + 4i \Im(a)
    }, \label{eq:z4} \\
  \z_5 &=
    \frac{1}{6} \sof[\big]{
        e^{i \varphi_1}
      - e^{i \varphi_2}
      - 2i \Im(a + \sqrt{3} c)
    }, \label{eq:z5} \\
  \z_6 &=
    \frac{1}{6} \sof[\big]{
        e^{i \varphi_1}
      - e^{i \varphi_2}
      - 2i \Im(a - \sqrt{3} c)
    }. \label{eq:z6}
\end{align}
Up to an overall global phase, this parametrizes precisely the set of coefficients for which the following matrix, cf. \cref{eq:Ls}, is unitary:
\begin{equation}
  \sum_{i=1}^6 \z_i L_i =
  \mx{
    \z_1 & \z_3 & \z_2 & \z_4 & \z_5 & \z_6 \\
    \z_2 & \z_1 & \z_3 & \z_6 & \z_4 & \z_5 \\
    \z_3 & \z_2 & \z_1 & \z_5 & \z_6 & \z_4 \\
    \z_4 & \z_6 & \z_5 & \z_1 & \z_2 & \z_3 \\
    \z_5 & \z_4 & \z_6 & \z_3 & \z_1 & \z_2 \\
    \z_6 & \z_5 & \z_4 & \z_2 & \z_3 & \z_1
  }.
  \label{eq:zL}
\end{equation}

\begin{remark}
If we insist, in addition to \cref{eq:z1,eq:z2,eq:z3,eq:z4,eq:z5,eq:z6}, that $\abs{z_1} = \dotsb = \abs{z_6} = 1/\sqrt{6}$, we get a $6 \times 6$ \emph{flat} unitary---this is also known as a \emph{complex Hadamard matrix}~\cite{TZ06}.  Such matrices are relevant to the MUB problem in six dimensions~\cite{BBE07}.  Using a computer, one can find a total of 72 discrete solutions $(z_1, \dotsc, z_6)$ under the flatness constraint.  Unfortunately, the corresponding unitaries appear to be equivalent to the $6 \times 6$ Fourier matrix.  Nevertheless, this method can in principle be used to find flat unitaries of size $\abs{G} \times \abs{G}$, for any finite group $G$.  It would be interesting to know whether this construction can yield anything beyond what is already known~\cite{TZ06}.
\end{remark}

\subsection{Imposing the independence constraints} \label{sect:Nice}

Recall that \cref{eq:Magic} involves terms with coefficients $\Re(\z_1 \bar{\z}_4)$, $\Re(\z_2 \bar{\z}_5)$, and $\Re(\z_3 \bar{\z}_6)$.  We would like to understand when these terms vanish (see \Q2 in \cref{sect:Contraction}), since this would ensure that the coefficients are independent of the states.

\begin{proposition}\label{prop:Res}
$\Re(\z_1 \bar{\z}_4) = \Re(\z_2 \bar{\z}_5) = \Re(\z_3 \bar{\z}_6) = 0$ if $\varphi_1 = - \varphi_2$.
\end{proposition}

\begin{proof}
Note from \cref{eq:z1,eq:z2,eq:z3,eq:z4,eq:z5,eq:z6} that (up to an overall constant of $1/6$) each coefficient $\z_i$ has one of the following two forms:
\begin{align}
  u &:= e^{i \varphi_1}
      + e^{i \varphi_2}
      + r \cos \alpha, \\
  v &:= e^{i \varphi_1}
      - e^{i \varphi_2}
      + ir \sin \alpha
\end{align}
for some $\varphi_1, \varphi_2, \alpha \in [0, 2\pi)$ and $r \geq 0$.  Furthermore, each $u$-type coefficient is paired up with a corresponding $v$-type coefficient.  A straightforward calculation gives
\begin{align}
  \Re(u \bar{v})
  ={}& \Re(u) \Re(v) + \Im(u) \Im(v) \\
  ={}& \of{\cos \varphi_1 + \cos \varphi_2 + r \cos \alpha}
       \of{\cos \varphi_1 - \cos \varphi_2} \\
    &+ \of{\sin \varphi_1 + \sin \varphi_2}
       \of{\sin \varphi_1 - \sin \varphi_2 + r \sin \alpha} \nonumber \\
  ={}& (\cos \varphi_1)^2 - (\cos \varphi_2)^2
     + \of{\cos \varphi_1 - \cos \varphi_2} r \cos \alpha \\
    &+ (\sin \varphi_1)^2 - (\sin \varphi_2)^2
     + \of{\sin \varphi_1 + \sin \varphi_2} r \sin \alpha \nonumber \\
  ={}& \of{\cos \varphi_1 - \cos \varphi_2} r \cos \alpha
     + \of{\sin \varphi_1 + \sin \varphi_2} r \sin \alpha \\
  ={}& r \of[\big]{ \cos(\alpha - \varphi_1) - \cos(\alpha + \varphi_2)}.
\end{align}
We can guarantee that $\Re(u \bar{v}) = 0$ irrespectively of the values of $r$ and $\alpha$ by choosing $-\varphi_1 = \varphi_2$.  This makes all three terms vanish simultaneously.
\end{proof}

Following \cref{prop:Res}, we define $\varphi := \varphi_1 = -\varphi_2$.  Then \cref{eq:z1,eq:z2,eq:z3,eq:z4,eq:z5,eq:z6} become:
\begin{align}
  \z_1 &= \frac{1}{3} \sof[\big]{ \cos \varphi + 2 \Re(a) },              \label{z1} \\
  \z_2 &= \frac{1}{3} \sof[\big]{ \cos \varphi -   \Re(a + \sqrt{3} c) }, \label{z2} \\
  \z_3 &= \frac{1}{3} \sof[\big]{ \cos \varphi -   \Re(a - \sqrt{3} c) }, \label{z3} \\
  \z_4 &= \frac{i}{3} \sof[\big]{ \sin \varphi + 2 \Im(a) },              \label{z4} \\
  \z_5 &= \frac{i}{3} \sof[\big]{ \sin \varphi -   \Im(a + \sqrt{3} c) }, \label{z5} \\
  \z_6 &= \frac{i}{3} \sof[\big]{ \sin \varphi -   \Im(a - \sqrt{3} c) }. \label{z6}
\end{align}
Here we can choose any $\varphi \in [0,2\pi)$ and $a, c \in \C$ such that $\abs{a}^2 + \abs{c}^2 = 1$.  The output state is then obtained by substituting \cref{z1,z2,z3,z4,z5,z6} in \cref{eq:Magic}.

Unfortunately, the parametrization in \cref{z1,z2,z3,z4,z5,z6} is somewhat cumbersome.  In addition, it has another, more serious drawback: it appears as if there are four degrees of freedom---one from $\varphi$ and three from $a$ and $c$.  While it is not obvious from \cref{z1,z2,z3,z4,z5,z6}, one of these degrees of freedom is redundant, since it has no affect on the output state.

\subsection{Alternative parametrizations} \label{sect:Parametrizations}

Due to the shortcomings just discussed, in this section we derive two alternative parametrizations that are much simpler and more insightful.  Our derivation is based on \cref{z1,z2,z3,z4,z5,z6}.  With an educated guess, however, the same alternative parametrizations can also be derived from scratch without invoking the irreps of $\S_3$ at all (see \cref{apx:From scratch}).

\subsubsection{Parametrization by \texorpdfstring{$\C^3$}{C3}} \label{sect:C3}

It is natural to pair up the coefficients $\z_i$ as follows:
\begin{align}
  q_1 &:= \z_1 + \z_4 = \frac{1}{3} \of{ e^{i\varphi} + 2a},             \nonumber \\
  q_2 &:= \z_2 + \z_5 = \frac{1}{3} \of{ e^{i\varphi} - a - \sqrt{3}c }, \label{eq:q123} \\
  q_3 &:= \z_3 + \z_6 = \frac{1}{3} \of{ e^{i\varphi} - a + \sqrt{3}c }. \nonumber
\end{align}
The coefficients $\z_i$ in terms of the new parameters $q_1, q_2, q_3 \in \C$ are expressed as follows:
\begin{equation}
  (\z_1, \z_2, \z_3, \z_4, \z_5, \z_6) :=
  (\Re q_1,  \Re q_2,  \Re q_3,
  i\Im q_1, i\Im q_2, i\Im q_3).
  \label{eq:zqs}
\end{equation}
With this parametrization, the output state $\rho$ from \cref{eq:Magic} looks as follows:
\begin{falign}
  \rho
  &= \abs{q_1}^2 \rho_1
   + \abs{q_2}^2 \rho_2 \label{q:1}
   + \abs{q_3}^2 \rho_3 \\[0.2cm]
  &+ \Im(q_1 \bar{q}_2) \, i [\rho_1, \rho_2] \nonumber \\
  &+ \Im(q_2 \bar{q}_3) \, i [\rho_2, \rho_3] \label{q:2}\\
  &+ \Im(q_3 \bar{q}_1) \, i [\rho_3, \rho_1] \nonumber \\[0.2cm]
  &+ \Re(q_1 \bar{q}_2) \of{\rho_2 \rho_3 \rho_1 + \rho_1 \rho_3 \rho_2} \nonumber \\
  &+ \Re(q_2 \bar{q}_3) \of{\rho_3 \rho_1 \rho_2 + \rho_2 \rho_1 \rho_3} \label{q:3} \\
  &+ \Re(q_3 \bar{q}_1) \of{\rho_1 \rho_2 \rho_3 + \rho_3 \rho_2 \rho_1} \nonumber
\end{falign}
where $q_1, q_2, q_3 \in \C$ are subject to the following constraints:
\begin{align}
  \abs{q_1}^2 + \abs{q_2}^2 + \abs{q_3}^2 &= 1, &
  \abs{q_1 + q_2 + q_3}^2 &= 1.
  \label{eq:qs}
\end{align}

To derive these constraints, we solve \cref{eq:q123} for the original parameters:
\begin{align}
  e^{i\varphi} &= q_1 + q_2 + q_3, &
  a &= q_1 - \frac{1}{2} (q_2 + q_3), &
  c &= \frac{\sqrt{3}}{2} (q_3 - q_2).
\end{align}
From the first equation we immediately get the second constraint in \cref{eq:qs}.  From the next two equations we get:
\begin{align}
  \abs{a}^2
   &= \abs{q_1}^2 + \frac{1}{4} \of[\big]{ \abs{q_2}^2 + \abs{q_3}^2 }
    - \Re \of{ q_1 \bar{q}_2 }
    - \Re \of{ q_3 \bar{q}_1 }
    + \frac{1}{2} \Re \of{ q_2 \bar{q}_3 }, \\
  \abs{c}^2
   &= \frac{3}{4} \of[\big]{ \abs{q_2}^2 + \abs{q_3}^2 }
    - \frac{3}{2} \Re \of{ q_2 \bar{q}_3 }.
\end{align}
Adding these together gives:
\begin{equation}
  1 = \abs{a}^2 + \abs{c}^2
    = \abs{q_1}^2 + \abs{q_2}^2 + \abs{q_3}^2
    - \Re \of{ q_1 \bar{q}_2 }
    - \Re \of{ q_2 \bar{q}_3 }
    - \Re \of{ q_3 \bar{q}_1 }.
  \label{eq:-Re}
\end{equation}
However, we already know from the second constraint in \cref{eq:qs} that
\begin{equation}
  1 = \abs{q_1 + q_2 + q_3}^2
    = \abs{q_1}^2 + \abs{q_2}^2 + \abs{q_3}^2
    + 2 \Re \of{q_1 \bar{q}_2}
    + 2 \Re \of{q_2 \bar{q}_3}
    + 2 \Re \of{q_3 \bar{q}_1}.
  \label{eq:2Re}
\end{equation}
Comparing \cref{eq:2Re,eq:-Re} we conclude that
\begin{equation}
  \Re \of{ q_1 \bar{q}_2 } +
  \Re \of{ q_2 \bar{q}_3 } +
  \Re \of{ q_3 \bar{q}_1 } = 0.
  \label{eq:qRes}
\end{equation}
Consequently, constraints~\eqref{eq:2Re} and~\eqref{eq:-Re} are equivalent to \cref{eq:qs}.  These constraints are also equivalent to unitarity of the three matrices in \cref{eq:Us}, assuming $\varphi_1 = -\varphi_2 =: \varphi$.  Indeed, if we substitute \cref{eq:q123} in \cref{eq:qs}, we recover $\abs{e^{i \varphi}}^2 = 1$ and $\abs{a}^2 + \abs{c}^2 = 1$.

Another way of writing constraints \eqref{eq:qs} is as follows.  If we let
\begin{align}
  \ket{q} &:= \mx{q_1 \\ q_2 \\ q_3}, &
  \ket{u} &:= \frac{1}{\sqrt{3}} \mx{1 \\ 1 \\ 1}
\end{align}
then \cref{eq:qs} is equivalent to
\begin{align}
  \abs{\braket{q}{q}}^2 &= 1, &
  \abs{\braket{q}{u}}^2 &= 1/3,
  \label{eq:Overlaps}
\end{align}
i.e., $\ket{q} \in \C^3$ is a unit vector that is \emph{mutually unbiased} to the uniform superposition $\ket{u}$.

One advantage of this parametrization is that it makes it more apparent that neither the constraints~\eqref{eq:Overlaps} nor the output state $\rho$ in \cref{q:1,q:2,q:3} depend on the global phase of $\ket{q}$; this feature was not obvious at all from the original parametrization in \cref{z1,z2,z3,z4,z5,z6}.  In fact, we can always adjust the global phase of $\ket{q}$ so that the $\abs{q_1 + q_2 + q_3}^2 = 1$ constraint from \cref{eq:qs} turns into $q_1 + q_2 + q_3 = 1$.  A simplified version of \cref{eq:qs} is then
\begin{align}
  \abs{q_1}^2 + \abs{q_2}^2 + \abs{q_3}^2 &= 1, &
  q_1 + q_2 + q_3 &= 1,
  \label{eq:qs'}
\end{align}
which makes it very clear that $\ket{q}$ has only three relevant degrees of freedom.  We will explore these constraints further in \cref{sect:Linkage} and relate them to the so-called four-bar linkage mechanism.

We can make some further observations:
\begin{itemize}
  \item Due to the first constraint in \cref{eq:qs'}, the first-order terms in \cref{q:1} form a convex combination of $\rho_1,\rho_2,\rho_3$.  This is analogous to \cref{eq:box} for $n = 2$.
  \item If $\ket{q}$ satisfies \cref{eq:qs'} then so does its complex conjugate $\ket{q}^*$.  Complex conjugation of $q_i$'s preserves \cref{q:1,q:3} but flips the signs of the second-order terms in \cref{q:2}.
  \item Because of \cref{eq:qRes}, the coefficients of the third-order terms in \cref{q:3} sum to zero.  Hence, the third-order terms themselves can be expressed as a linear combination of double commutators (see \cref{sect:Double commutators} for more details).
  \item If the input states $\rho_i$ and the coefficients $q_i$ are permuted according to the same permutation, the output state in \cref{q:1,q:2,q:3} remains invariant.  For cyclic permutations, this is evident form symmetry; it is straightforward to check for transpositions.
\end{itemize}

\subsubsection{Parametrization by a probability distribution and phases} \label{sect:pd}

As we just noted, \cref{eq:qs'} implies that $(\abs{q_1}^2, \abs{q_2}^2, \abs{q_3}^2)$ is a probability distribution.  We can highlight this by setting
\begin{equation}
  q_k := e^{i \phi_k} \sqrt{p_k},
\end{equation}
for some probability distribution $(p_1, p_2, p_3)$ and some phases $\phi_1, \phi_2, \phi_3 \in [0, 2\pi)$.  Note that
\begin{align}
  \Re(q_i \bar{q}_j) &= \sqrt{p_i p_j} \cos (\phi_i - \phi_j), \\
  \Im(q_i \bar{q}_j) &= \sqrt{p_i p_j} \sin (\phi_i - \phi_j).
\end{align}
If we further denote the differences between consecutive phases by
\begin{align}
  \delta_{12} &:= \phi_1 - \phi_2, \label{eq:d12} \\
  \delta_{23} &:= \phi_2 - \phi_3, \label{eq:d23} \\
  \delta_{31} &:= \phi_3 - \phi_1, \label{eq:d31}
\end{align}
we can rewrite \cref{q:1,q:2,q:3} as follows:
\begin{falign}
  \rho
  &= p_1 \rho_1
   + p_2 \rho_2 \label{p:1}
   + p_3 \rho_3 \\[0.2cm]
  &+ \sqrt{p_1 p_2} \, \sin \delta_{12} \, i [\rho_1, \rho_2] \nonumber \\
  &+ \sqrt{p_2 p_3} \, \sin \delta_{23} \, i [\rho_2, \rho_3] \label{p:2} \\
  &+ \sqrt{p_3 p_1} \, \sin \delta_{31} \, i [\rho_3, \rho_1] \nonumber \\[0.2cm]
  &+ \sqrt{p_1 p_2} \, \cos \delta_{12} \, \of{\rho_2 \rho_3 \rho_1 + \rho_1 \rho_3 \rho_2} \nonumber \\
  &+ \sqrt{p_2 p_3} \, \cos \delta_{23} \, \of{\rho_3 \rho_1 \rho_2 + \rho_2 \rho_1 \rho_3} \label{p:3} \\
  &+ \sqrt{p_3 p_1} \, \cos \delta_{31} \, \of{\rho_1 \rho_2 \rho_3 + \rho_3 \rho_2 \rho_1}.\nonumber
\end{falign}

Parameters $p_i$ and $\delta_{ij}$ are subject to the following constraints: $(p_1, p_2, p_3)$ is a probability distribution (i.e., $p_i \geq 0$ and $p_1 + p_2 + p_3 = 1$) and the angles $\delta_{ij}$ satisfy
\begin{align}
  \delta_{12} + \delta_{23} + \delta_{31} &= 0, &
  \sqrt{p_1 p_2} \, \cos \delta_{12}
+ \sqrt{p_2 p_3} \, \cos \delta_{23}
+ \sqrt{p_3 p_1} \, \cos \delta_{31} &= 0.
  \label{eq:deltas}
\end{align}
The first constraint is apparent from \cref{eq:d12,eq:d23,eq:d31}, while the second constraint is equivalent to \cref{eq:qRes}.  In total, there are three degrees of freedom: two for the distribution $(p_1,p_2,p_3)$ and one for the angles $(\delta_{12}, \delta_{23}, \delta_{31})$.  If the distribution $(p_1,p_2,p_3)$ is fixed (and not deterministic), the constraints in \cref{eq:deltas} yield a one-parameter family of angles $(\delta_{12}, \delta_{23}, \delta_{31})$.  This is qualitatively different from the $n = 2$ case where the coefficients of the first-order terms are completely determined by $\lambda$, see \cref{eq:box}.  Once the parameter $\lambda$ is fixed, only a discrete degree of freedom remains corresponding to the sign in front of the commutator, see \cref{sect:General U}.

\section{Further observations}

In this section we highlight some further features of the operation that combines three states.  In particular, we show that the third-order terms can be expressed as a linear combination of double commutators.  We also show that the twelve nested compositions discussed in \cref{sect:12 ways} are special cases of the general operation.

\subsection{Double commutators} \label{sect:Double commutators}

Let us elaborate more on the meaning of the constraint \eqref{eq:qRes}.  As mentioned earlier, it has to do with \emph{double commutators}, i.e., expressions of the form $[1,[2,3]]$.\footnote{We write $i$ instead of $\rho_i$ for brevity.}  In what follows, we do not specify the states $\rho_i$ but rather treat them as abstract non-commutative variables.  With this convention, for example, $\rho_1$ and $\rho_2$ are always considered to be linearly independent.

Note that there are $6$ ways of ordering three states and $2$ ways of putting brackets, so there are twelve double commutators in total.  However, many of them are identical (such as $[1,[2,3]]$ and $[[3,2],1]$) or differ only by a sign (such as $[1,[2,3]]$ and $[1,[3,2]]$).  Furthermore, we know from the \emph{Jacobi identity} that
\begin{equation}
  [1, [2, 3]]
+ [2, [3, 1]]
+ [3, [1, 2]] = 0.
  \label{eq:Jacobi}
\end{equation}
This leaves us with only two linearly independent double commutators.  Somewhat arbitrarily, we can choose them as $[1,[2,3]]$ and $[[1,2],3]$.  After expanding both of them, we get the following coefficients in front of the six different products of the states $1$, $2$, $3$:
\begin{equation}
\begin{array}{r|r|r|r|r|r|r}
                                                      & 123 & 132 & 213 & 231 & 312 & 321 \\ \hline
  \sof{1,\sof{2,3}}                                   &  1  & -1  &  0  & -1  &  0  &  1  \\
  \sof{\sof{1,2},3}                                   &  1  &  0  & -1  &  0  & -1  &  1  \\
  x \, i\sof{1,i\sof{2,3}} + y \, i\sof{i\sof{1,2},3} &  z  &  x  &  y  &  x  &  y  &  z
\end{array}
\label{eq:table}
\end{equation}
where the last row is a linear combination of the first two rows and $z := - x - y$.  In other words,
\begin{equation}
    x (231 + 132)
  + y (312 + 213)
  + z (123 + 321)
  = x \, i\sof{1,i\sof{2,3}} + y \, i\sof{i\sof{1,2},3}
  \label{eq:xyz}
\end{equation}
whenever $x + y + z = 0$.

We can use this to simply the third-order terms in \cref{q:3,p:3}.
Indeed, since the coefficients in~\eqref{q:3} sum to zero, see \cref{eq:qRes}, we can rewrite~\eqref{q:3} as a linear combination of two double commutators:
\begin{falign}
  &+ \Re(q_1 \bar{q}_2) \, i\sof{\rho_1,i\sof{\rho_2,\rho_3}} \\[3pt]
  &+ \Re(q_2 \bar{q}_3) \, i\sof{i\sof{\rho_1,\rho_2},\rho_3}.\nonumber
\end{falign}
Since expression~\eqref{p:3} is symmetric, one can also use the above expression with all subscripts cyclically shifted by one or by two (i.e., according to $1 \to 2 \to 3 \to 1$ or its inverse).  Alternatively, if we use the parametrization in terms of $p_i$ and $\delta_{ij}$ from \cref{sect:pd}, we can restate~\eqref{p:3} as
\begin{falign}
 &+ \sqrt{p_1 p_2} \, \cos \delta_{12} \, i\sof{\rho_1,i\sof{\rho_2,\rho_3}} \label{eq:p order 3} \\[3pt]
 &+ \sqrt{p_2 p_3} \, \cos \delta_{23} \, i\sof{i\sof{\rho_1,\rho_2},\rho_3}.\nonumber
\end{falign}
Again, because of the second condition in \cref{eq:deltas}, this expression is invariant under cyclic shifts.

\subsection{Relation to nested expressions}

It is interesting to know whether the ternary operation can reproduce the 12 nested expressions discussed in \cref{sect:12 ways} as special cases.  More precisely, we would like to know whether, for fixed distribution $(p_1, p_2, p_3)$, the one-parameter family of states described by \cref{p:1,p:2,p:3}, under constraints~\eqref{eq:deltas}, contains all 12 nested expressions with first-order terms $p_1 \rho_1 + p_2 \rho_2 + p_3 \rho_3$.

For example, consider nested expressions of the form $\rho_1 \bs_a \of{\rho_2 \bs_{a'} \rho_3}$ where each~$\bs$ is either~$\bp$ or~$\bm$, see \cref{eq:+,eq:-}.  Recall from \cref{eq:d1,eq:d2} that
\begin{align}
  \rho_1 \bs_a \of{\rho_2 \bs_{a'} \rho_3}
    =          a \rho_1 + (1-a) &  \of[\Big]{a' \rho_2 + (1-a') \rho_3 + (-1)^{s'} \sqrt{a'(1-a')} \, i[\rho_3, \rho_2]}             \label{eq:l1} \\
    + (-1)^s \sqrt{a(1-a)} \, i & \sof[\Big]{a' \rho_2 + (1-a') \rho_3 + (-1)^{s'} \sqrt{a'(1-a')} \, i[\rho_3, \rho_2], \; \rho_1}, \label{eq:l2}
\end{align}
where $s,s' \in \set{0,1}$ depend on the signs of the two~$\bs$ operations ($0$ stands for~$\bp$ while~$1$ stands for~$\bm$).  Importantly, this expression has only one double commutator, namely $i\sof{\rho_1,i\sof{\rho_2,\rho_3}}$, while a general expression involves two, see~\cref{eq:p order 3}.

To get only one double commutator, we can set $\cos \delta_{ij} = 0$ for some $ij \in \set{12,23,31}$.  However, we must also be able to solve \cref{eq:deltas} for the remaining two angles.  For example, to get only $i\sof{\rho_1,i\sof{\rho_2,\rho_3}}$, we can set $\cos \delta_{23} = 0$ and solve \cref{eq:deltas} for $\delta_{12}$ and $\delta_{31}$.  Similarly, $\cos \delta_{31} = 0$ would yield $i[\rho_2,i[\rho_3,\rho_1]]$ via \cref{eq:deltas} and the Jacoby identity.  The following lemma establishes that we can indeed get all 12 nested expressions in this way.  Recall that throughout this section we treat $\rho_i$ as abstract non-commutative variables.

\begin{lemma}\label{lem:Nested}
Let $\rho$ be given by \cref{p:1,p:2,p:3} and assume $p_1,p_2,p_3 \neq 0$.  Then $\rho$ admits a nested expression if and only if $\cos \delta_{ij} = 0$ for some $ij \in \set{12,23,31}$.
\end{lemma}

\begin{proof}
We have already proved the forward implication.  Indeed, nested expressions have only one double commutator, so we get $\cos \delta_{ij} = 0$ for some $ij \in \set{12,23,31}$.

For the reverse implication, we assume $\cos \delta_{23} = 0$ (the other cases follow similarly due to symmetry).  Under this assumption, let us argue that \cref{eq:deltas} has only four discrete solutions, and that these solutions correspond to the four nested expressions
\begin{equation}
  \rho_1 \bs_{p_1} \of[\Big]{ \rho_2 \bs_{\frac{p_2}{p_2 + p_3}} \rho_3 }
  \label{eq:123}
\end{equation}
with arbitrary signs $s,s' \in \set{0,1}$ for the two $\bs$ operations.

First, we can write
\begin{equation}
  \cos \delta_{23} = 0 \qquad\text{and}\qquad
  \sin \delta_{23} = -(-1)^{s'}
  \label{d23}
\end{equation}
for some $s' \in \set{0,1}$.  In other words, $\delta_{23} := -(-1)^{s'} \pi/2$.  From the first half of \cref{eq:deltas}, $\delta_{12} = - \delta_{23} - \delta_{31} = (-1)^{s'} \pi / 2 - \delta_{31}$.  Hence
\begin{equation}
  \cos \delta_{12} = (-1)^{s'} \sin \delta_{31} \qquad\text{and}\qquad
  \sin \delta_{12} = (-1)^{s'} \cos \delta_{31}.
  \label{eq:cs}
\end{equation}
From the second half of \cref{eq:deltas}, $\sqrt{p_2} \cos \delta_{12} + \sqrt{p_3} \cos \delta_{31} = 0$.  If we substitute $\cos \delta_{12}$ from \cref{eq:cs}, this becomes $\sqrt{p_2} \, (-1)^{s'} \sin \delta_{31} + \sqrt{p_3} \cos \delta_{31} = 0$ and we get
\begin{equation}
  \sin \delta_{31} =   (-1)^s      \sqrt{\frac{p_3}{p_2 + p_3}} \qquad\text{and}\qquad
  \cos \delta_{31} = - (-1)^{s+s'} \sqrt{\frac{p_2}{p_2 + p_3}}
  \label{d31}
\end{equation}
for some $s \in \set{0,1}$.  We can then find $\delta_{12}$ by substituting \cref{d31} back in \cref{eq:cs}:
\begin{equation}
  \sin \delta_{12} = - (-1)^s      \sqrt{\frac{p_2}{p_2 + p_3}} \qquad\text{and}\qquad
  \cos \delta_{12} =   (-1)^{s+s'} \sqrt{\frac{p_3}{p_2 + p_3}}.
  \label{d12}
\end{equation}
\Cref{d23,d31,d12} with $s,s' \in \set{0,1}$ give us the four solutions.  It remains to argue that these solutions produce the four states in \cref{eq:123}.

We begin by restating \cref{p:1,p:2,p:3} when $\cos \delta_{23} = 0$:
\begin{align}
  \rho
  &= p_1 \rho_1
   + p_2 \rho_2 \label{pp:1}
   + p_3 \rho_3 \\[0.2cm]
  &+ \sqrt{p_1 p_2} \, \sin \delta_{12} \, i [\rho_1, \rho_2] \nonumber \\
  &+ \sqrt{p_2 p_3} \, \sin \delta_{23} \, i [\rho_2, \rho_3] \label{pp:2} \\
  &+ \sqrt{p_3 p_1} \, \sin \delta_{31} \, i [\rho_3, \rho_1] \nonumber \\[0.2cm]
  &+ \sqrt{p_1 p_2} \, \cos \delta_{12} \, i \sof{\rho_1,i\sof{\rho_2,\rho_3}}. \label{pp:3}
\end{align}
Let us group the terms together to make this appear more similar to \cref{eq:l1,eq:l2}:
\begin{align}
  \rho
  &= p_1 \rho_1 +
     \of[\big]{
       p_2 \rho_2
     + p_3 \rho_3
     - \sqrt{p_2 p_3} \, \sin \delta_{23} \, i [\rho_3, \rho_2]
     } \\
  &+ i
     \sof[\big]{
     - \sqrt{p_1 p_2} \, \sin \delta_{12} \, \rho_2
     + \sqrt{p_3 p_1} \, \sin \delta_{31} \, \rho_3
     + \sqrt{p_1 p_2} \, \cos \delta_{12} \, i\sof{\rho_3,\rho_2}
     , \; \rho_1}.
\end{align}
If we pull out the desired prefactors, we get
\begin{align}
  \rho
   = p_1 \rho_1 &+
     \of{1 - p_1}
     \of[\bigg]{
       \frac{p_2}{p_2 + p_3} \, \rho_2
     + \frac{p_3}{p_2 + p_3} \, \rho_3
     - \frac{\sqrt{p_2 p_3}}{p_2 + p_3} \sin \delta_{23} \, i [\rho_3, \rho_2]
     } \\
  &+ \sqrt{p_1 (1 - p_1)} \, i
     \sof[\bigg]{
     - \frac{\sqrt{p_2} \, \sin \delta_{12}}{\sqrt{p_2 + p_3}} \, \rho_2
     + \frac{\sqrt{p_3} \, \sin \delta_{31}}{\sqrt{p_2 + p_3}} \, \rho_3
     + \frac{\sqrt{p_2} \, \cos \delta_{12}}{\sqrt{p_2 + p_3}} \, i\sof{\rho_3,\rho_2}
     , \; \rho_1}.
\end{align}
This becomes one of the four states in \cref{eq:123} once we substitute the values of all $\sin \delta_{ij}$ and $\cos \delta_{ij}$ from \cref{d23,d31,d12}.
\end{proof}

\subsection{The four-bar linkage and its number of orbits} \label{sect:Linkage}

Recall from \cref{q:1,q:2,q:3} in \cref{sect:C3} that the output state can be parametrized by $q_1, q_2, q_3 \in \C$, subject to \cref{eq:qs'}, i.e.,
\begin{align}
  \abs{q_1}^2 + \abs{q_2}^2 + \abs{q_3}^2 &= 1, &
  q_1 + q_2 + q_3 &= 1.
  \label{eq:qqq}
\end{align}
If we fix the absolute values $\abs{q_1}$, $\abs{q_2}$, $\abs{q_3}$ (and hence also the coefficients of the first-order terms in the output state), the set of possible solutions $(q_1, q_2, q_3)$ coincides with the configuration space of a \emph{four-bar linkage}, with bars of lengths $\abs{q_1}, \abs{q_2}, \abs{q_3}$, and a fourth (immobile) bar of length $1$ (see \cref{fig:Linkage}).  This simple mechanism played an important role in the invention of the steam engine and the bicycle in the 18th and 19th century, respectively, and it has found many other practical applications since.  Interestingly, it also occurs in nature, such as in the human knee joint and the jaw of a parrotfish.

\begin{figure}
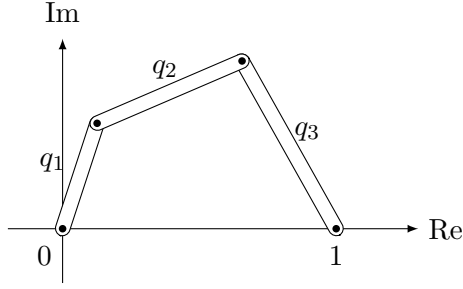

\centering
\linkage
\caption{\label{fig:Linkage}A four-bar linkage with bars $q_1, q_2, q_3 \in \C$ such that $q_1 + q_2 + q_3 = 1$.  The fourth bar (not shown) is immobile and corresponds to the interval $[0,1]$ on the real axis.}
\end{figure}

The four-bar linkage can have several different modes of operation, depending on the lengths of its bars (see \cite[p.~111]{Grashof} or \cite[Sections~2.3 and~2.4]{Linkages} for more details).  These modes are classified into two broad classes, based on the so-called \emph{Grashof condition}~\cite{Grashof}:
\begin{equation}
  a + d < b + c,
  \label{eq:Grashof}
\end{equation}
where $0 \leq a \leq b \leq c \leq d$ are the lengths of the four bars.  Under this condition, for example, the shortest bar $a$ can rotate fully.  More importantly, under this condition the mechanism has \emph{two} disjoint orbits, while it has a single orbit otherwise.\footnote{If \cref{eq:Grashof} holds with equality, the two disjoint orbits touch, thus merging into a single orbit.}  The two orbits are related to each other by a reflection around the real axis (or by complex conjugation of all $q_i$'s).  They also correspond to two different ways of assembling the mechanism.\footnote{By changing the angle between, say, $q_1$ and $q_2$ from reflex to non-reflex (and vice versa), one can obtain two configurations that belong to different orbits and are otherwise not reachable from one another.}

In the context of combining quantum states, the longest bar $d$ is always fixed since otherwise we would violate the first condition in \cref{eq:qqq}.  In other words, we have $d = 1$ and $a^2 + b^2 + c^2 = 1$.  Hence the number of orbits formed by possible output states, when the coefficients $\abs{q_1}$, $\abs{q_2}$, $\abs{q_3}$ in \cref{q:1} are fixed, can be determined as follows.

\begin{proposition}\label{prop:Orbits}
Let $a,b,c,d$ be as above, i.e., $0 \leq a \leq b \leq c \leq d = 1$ and $a^2 + b^2 + c^2 = 1$.  Then the corresponding four-bar linkage has one orbit, unless $b > b_0(c)$, in which case it has two; here
\begin{equation}
  b_0(c) := \frac{1}{2} \of*{1 - c + \sqrt{1 + (2 - 3c) c}}.
\end{equation}
In particular, if $c \leq 2/3$ then $b_0(c) \geq c \geq b$ and there is just one orbit irrespective of the value of $b$.
\end{proposition}

\begin{proof}
If $c \leq 2/3$ then $b + c \leq 2c \leq 4/3$ and $a^2 = 1-b^2-c^2 \geq 1-2c^2 \geq 1 - 8/9 = 1/9$, so $1 + a \geq 4/3 \geq b + c$ violates the Grashof condition~\eqref{eq:Grashof} and there is only one orbit irrespective of $b$.  If $c > 2/3$, we find the critical value $b_0(c)$ by solving $a+1=b+c$ and $a^2+b^2+c^2=1$.
\end{proof}

Intuitively, if $b$ and $c$ are both sufficiently large, there are two orbits.  In particular, in the extreme case when $a = 0$ we have $b = \sqrt{1 - c^2} > b_0(c)$ (assuming $c \neq 1$), meaning that there are two orbits for all choices of $b$ (except for $b = 0$, of course).  This is consistent with the fact that the $n=2$ case has two discrete solutions, see \cref{eq:+,eq:-}.

\subsection{The uniform combination} \label{sect:Uniform}

Let us consider the special case when $p_1 = p_2 = p_3 = 1/3$ (or when $\abs{q_1} = \abs{q_2} = \abs{q_3} = 1/\sqrt{3}$), i.e., the \emph{uniform combination}.  Using \cref{p:1,p:2,p:3}, the output state $\rho$ can be written as
\begin{falign}
  3 \rho
  &= \rho_1 + \rho_2 + \rho_3 \label{1/3:1} \\[0.2cm]
  &+ \sin \delta_{12} \, i [\rho_1, \rho_2] \nonumber \\
  &+ \sin \delta_{23} \, i [\rho_2, \rho_3] \label{1/3:2} \\
  &+ \sin \delta_{31} \, i [\rho_3, \rho_1] \nonumber \\[0.2cm]
  &+ \cos \delta_{12} \, \of{\rho_2 \rho_3 \rho_1 + \rho_1 \rho_3 \rho_2} \nonumber \\
  &+ \cos \delta_{23} \, \of{\rho_3 \rho_1 \rho_2 + \rho_2 \rho_1 \rho_3} \label{1/3:3} \\
  &+ \cos \delta_{31} \, \of{\rho_1 \rho_2 \rho_3 + \rho_3 \rho_2 \rho_1},\nonumber
\end{falign}
where the angles $\delta_{ij}$ are subject to the following relations, see \cref{eq:deltas}:
\begin{align}
  \delta_{12} + \delta_{23} + \delta_{31} &= 0, &
  \cos \delta_{12} + \cos \delta_{23} + \cos \delta_{31} &= 0.
  \label{eq:deltas'}
\end{align}

The corresponding four-bar linkage in this case has $a = b = c = 1/\sqrt{3}$, i.e., all non-stationary bars are of the same length.  \Cref{prop:Orbits} implies that such mechanism has a single orbit.  In other words, \cref{1/3:1,1/3:2,1/3:3} describe a single one-parameter orbit of states.  As proved in \cref{lem:Nested}, this orbit contains all twelve nested combinations (they are illustrated in \cref{fig:Clock}, together with the corresponding configurations of the four-bar linkage).

\begin{figure}
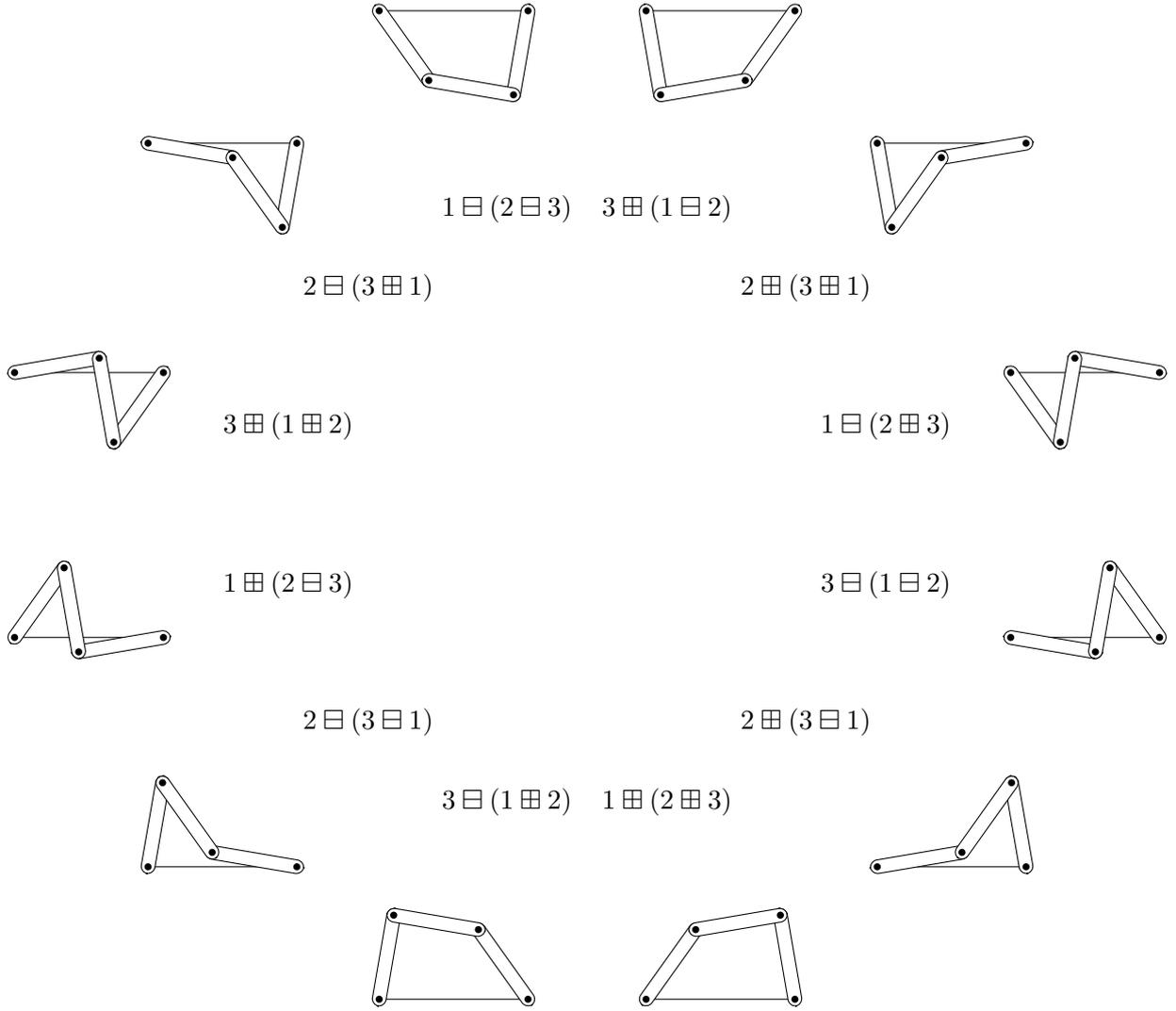

\centering
\clock
\caption{\label{fig:Clock}The orbit of a four-bar linkage with three bars of length $1/\sqrt{3}$ (and a fixed bar of length $1$) parametrizes all possible ways of uniformly combining three states.  Depicted are those linkage configurations where two of the moving bars are orthogonal---these configurations correspond to twelve \emph{nested} uniform combinations of three states (the inner and the outer $\bs$ operation in the nested combination have parameters $1/2$ and $1/3$, respectively).  The sign of the angle $\pm \pi/2$ between the two orthogonal bars determines the sign of the inner operation.  The sign of the outer operation is determined by the sign of the angle between the remaining bar and either of the other two bars.}
\end{figure}

Recall from \cref{lem:Nested} that nested combinations correspond to $\cos \delta_{ij} = 0$, for some $ij \in \set{12,23,31}$.  In terms of the complex parameters $q_i$ (see \cref{sect:C3}), this is equivalent to $\Re(q_i \bar{q}_j) = 0$, meaning that $q_i$ and $q_j$ are orthogonal as vectors in the complex plane.  Indeed, \cref{fig:Clock} illustrates exactly those configurations of the four-bar linkage where two of the three bars are orthogonal.  Moreover, the sign of the angle $\arg(q_i \bar{q}_j) = \pm \pi/2$ between $q_i$ and $q_j$ determines the sign of the inner $\bs$ operation in the nested combination.  The sign of the outer $\bs$ operation can be determined from the sign of the angle the remaining bar forms with either of the other two bars (both signs coincide).

As a side node, if the parameters $q_i$ are such that the four-bar linkage has two orbits (see \cref{prop:Orbits}), the output state cannot continuously pass through all twelve nested combinations, but only through those six for which either $\arg q_1 > 0$ or $\arg q_1 < 0$, depending on the initial configuration of the linkage.  To pass from one orbit to the other, one can simultaneously take the complex conjugate of all three parameters $q_i$.  In the extreme case when one of the $q_i$'s is zero, the two orbits degenerate to two points.  If two of the $q_i$'s are zero, the two points merge into one.

\begin{example}[Uniform combination of three mutually unbiased qubit states]
Let $I, \sigma_x, \sigma_y, \sigma_z$ be the Pauli matrices and $\rho(x,y,z) := \frac{1}{2} \of{I + x \sigma_x + y \sigma_y + z \sigma_z}$ be an arbitrary single-qubit state.  Let
$\rho_1 := \rho(1,0,0)$,
$\rho_2 := \rho(0,1,0)$,
$\rho_3 := \rho(0,0,1)$,
and $p_1 = p_2 = p_3 = 1/3$.
(Note that $\rho_1 = \proj{+}$, $\rho_2 = \proj{{+i}}$, $\rho_3 = \proj{0}$, i.e., these are mutually unbiased pure states pointing along the three axes of the Bloch sphere.)  According to \cref{1/3:1,1/3:2,1/3:3}, the combined state is given by
\begin{equation}
  \rho \of[\bigg]{
    \frac{1}{3} \of{1 - \sin \delta_{23}},
    \frac{1}{3} \of{1 - \sin \delta_{31}},
    \frac{1}{3} \of{1 - \sin \delta_{12}}
  },
\end{equation}
where the angles $\delta_{ij}$ are subject to relations \eqref{eq:deltas'}:
\begin{align}
  \delta_{12} + \delta_{23} + \delta_{31} &= 0, &
  \cos \delta_{12} + \cos \delta_{23} + \cos \delta_{31} &= 0.
\end{align}
The resulting one-parameter orbit is shown in \cref{fig:Orbit}.  According to \cref{lem:Nested}, this orbit contains all twelve nested uniform combinations of $\rho_1, \rho_2, \rho_3$, also illustrated in \cref{fig:Orbit}.
\end{example}

\begin{figure}
\centering
\begin{tikzpicture}[>=latex,
    line width = 1pt,
    arr/.style = {},
    pt/.style = {circle, fill = black, inner sep = 1.5pt}
  ]
\input{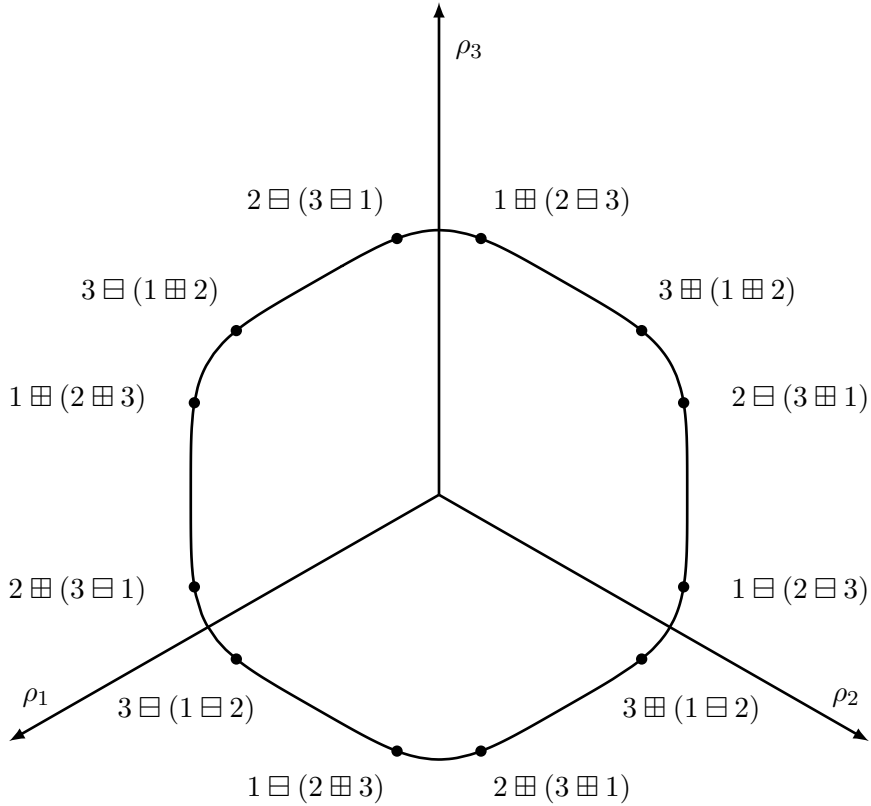}
\end{tikzpicture}
\caption{\label{fig:Orbit}The one-parameter orbit in the Bloch sphere corresponding to all uniform combinations of three mutually unbiased pure qubit states (they correspond to the three axes of the Bloch sphere).  This orbit contains all twelve nested combinations of the three states as special cases (the inner and the outer $\bs$ operations have parameters $1/2$ and $1/3$, respectively).  The orbit in this figure is rotated relative to the one in \cref{fig:Clock} by a certain angle.}
\end{figure}

\section{Open problems}

The main open problem is generalizing \cref{thm:EPI} (originally from~\cite{ADO16}) to three states.  Here is a formal statement of this conjecture.

\begin{conjecture*}[ADO inequality for three states]
For any concave and symmetric function $f: \D{d} \to \R$, any states $\rho_1, \rho_2, \rho_3 \in \D{d}$, and any probability distribution $(p_1, p_2, p_3)$,
\begin{equation}
  f(\rho)
  \geq
    p_1 f(\rho_1)
  + p_2 f(\rho_2)
  + p_3 f(\rho_3)
\end{equation}
where $\rho$ is given by \cref{p:1,p:2,p:3} with $p_i$ and $\delta_{ij}$ subject to \cref{eq:deltas}.
\end{conjecture*}

It would also be interesting to understand how an arbitrary number of states can be combined.  Towards this goal, the two main steps are:
\begin{enumerate}
  \item \emph{Finding a generalization of \cref{eq:Magic}.}  For general $n$, the expression of the output state $\rho$ has $(n!)^2$ terms, so we need a more efficient way of contracting tensor diagrams to compute it.
  \item \emph{Answering \Q2 for any $n$.}  While for $n=3$ it was sufficient to adjust the global phases of the unitaries $U_k$ in \cref{eq:Us}, for general $n$ it is not clear at all how to turn the first-order terms of $\rho$ into a convex combination of $\rho_i$, with coefficients depending only on the parameters $\z_{\pi}$ but not the input states $\rho_i$ themselves.
\end{enumerate}
It is also worthwhile investigating when higher-order terms of $\rho$ can be written as a linear combination of nested commutators.  Perhaps, as suggested by the $n=3$ case, better understanding of the second problem might make it possible to deal with both problems simultaneously, since many terms are likely to drop out simultaneously.

\section*{Acknowledgements}

I would like to thank Johannes Bausch, Andrew Childs, Toby Cubitt, Nilanjana Datta, Stephen Jordan, and Will Matthews for very helpful discussions.  This work was supported by the European Union under project QALGO (Grant Agreement No. 600700) and by a Leverhulme Trust Early Career Fellowship (ECF-2015-256).  In loving memory of L.O.V.\ Coffeebean.


\bibliographystyle{alphaurl}
\newcommand{\etalchar}[1]{$^{#1}$}

\appendix

\section{Deriving the parametrization from scratch} \label{apx:From scratch}

In this appendix we derive from scratch (namely, without using the irreps of $\S_3$) the parametrization of $\rho$ obtained in \cref{sect:C3}.  The only assumption that goes into our derivation is that $\z_1,\z_2,\z_3$ are real while $\z_4,\z_5,\z_6$ are imaginary---this is something we observed in \cref{sect:Nice}, \cref{z1,z2,z3,z4,z5,z6}.  This assumption is in fact sufficient for deriving \cref{q:1,q:2,q:3} and the constraints in \cref{eq:qs}.

Following \cref{eq:zqs}, we choose the coefficients $\z_i$ as
\begin{equation}
  (\z_1, \z_2, \z_3, \z_4, \z_5, \z_6) := (a_1, a_2, a_3, ib_1, ib_2, ib_3)
  \label{eq:ab}
\end{equation}
for some $a_1, a_2, a_3, b_1, b_2, b_3 \in \R$.  Without any additional constraints, this can potentially capture only more than \cref{z1,z2,z3,z4,z5,z6}.  At the same time, $\z_i$ chosen according to \cref{eq:ab} still satisfy the constraints imposed in \cref{prop:Res}.  In other words, this choice automatically takes care of \Q2.

Let us rewrite the output state $\rho$ from \cref{eq:Magic} in terms of the new parameters $a_i$ and $b_i$:
\begin{falign}
  \rho
  {}={}& \of[\big]{ a_1^2 + b_1^2 } \rho_1 \nonumber   \\
  {}+{}& \of[\big]{ a_2^2 + b_2^2 } \rho_2 \label{a:1} \\
  {}+{}& \of[\big]{ a_3^2 + b_3^2 } \rho_3 \nonumber   \\[0.2cm]
  {}+{}& \of[\big]{ a_1 b_2 - a_2 b_1 } i [\rho_2, \rho_1] \nonumber   \\
  {}+{}& \of[\big]{ a_2 b_3 - a_3 b_2 } i [\rho_3, \rho_2] \label{a:2} \\
  {}+{}& \of[\big]{ a_3 b_1 - a_1 b_3 } i [\rho_1, \rho_3] \nonumber   \\[0.2cm]
  {}+{}& \of[\big]{ a_1 a_2 + b_1 b_2 } \of{\rho_2 \rho_3 \rho_1 + \rho_1 \rho_3 \rho_2} \nonumber   \\
  {}+{}& \of[\big]{ a_2 a_3 + b_2 b_3 } \of{\rho_3 \rho_1 \rho_2 + \rho_2 \rho_1 \rho_3} \label{a:3} \\
  {}+{}& \of[\big]{ a_3 a_1 + b_3 b_1 } \of{\rho_1 \rho_2 \rho_3 + \rho_3 \rho_2 \rho_1}.\nonumber
\end{falign}
Clearly, this is identical to \cref{q:1,q:2,q:3} if we let $q_k := a_k + i b_k$.

Let us now work backwards to find what extra constraints should be imposed on the coefficients $a_i$ and $b_i$ to satisfy the unitarity requirement \Q1.  Recall from \cref{lem:Regular rep} that we want the following matrix, see \cref{eq:zL}, to be unitary:
\begin{equation}
  \sum_{k=1}^6 \z_k L_k =
  \mx{
    a_1 & a_3 & a_2 & ib_1 & ib_2 & ib_3 \\
    a_2 & a_1 & a_3 & ib_3 & ib_1 & ib_2 \\
    a_3 & a_2 & a_1 & ib_2 & ib_3 & ib_1 \\
    ib_1 & ib_3 & ib_2 & a_1 & a_2 & a_3 \\
    ib_2 & ib_1 & ib_3 & a_3 & a_1 & a_2 \\
    ib_3 & ib_2 & ib_1 & a_2 & a_3 & a_1
  } =
  \mx{A & iB\tp \\ iB & A\tp} =: U
\end{equation}
where
\begin{align}
  A &:=
  \mx{a_1 & a_3 & a_2 \\
      a_2 & a_1 & a_3 \\
      a_3 & a_2 & a_1 }, &
  B &:=
  \mx{b_1 & b_3 & b_2 \\
      b_2 & b_1 & b_3 \\
      b_3 & b_2 & b_1 }.
\end{align}
We can write the unitarity condition as
\begin{equation}
  U U\ct
  = \mx{A & iB\tp \\ iB & A\tp}
    \mx{A\tp & -iB\tp \\ -iB & A}
  = \mx{A A\tp + B\tp B & i[B\tp,A] \\
        i[B,A\tp] & A\tp A + B B\tp}
  = I.
\end{equation}
Note that $[B\tp,A] = 0$ holds automatically, so the remaining constraints follow solely from the diagonal blocks: $A A\tp + B\tp B = A\tp A + B B\tp = I$.  These constraints are:
\begin{align}
  a_1^2 + a_2^2 + a_3^2 + b_1^2 + b_2^2 + b_3^2 &= 1, \label{eq:one} \\
  a_1 a_2 + a_2 a_3 + a_3 a_1 + b_1 b_2 + b_2 b_3 + b_3 b_1 &= 0. \label{eq:zero}
\end{align}
It is not hard to see that these constraints are equivalent to \cref{eq:qs}.

\end{document}